%% file: main.tex
\pdfoutput=1
\documentclass[11pt]{article}
\usepackage[utf8]{inputenc} %

\usepackage[
backend=biber,
style=alphabetic,
maxbibnames=15,
maxalphanames=10,
minalphanames=6,
giveninits=true,
doi=false,
isbn=false,
url=false,
eprint=false,
backref=true,
]{biblatex}
\usepackage[T1]{fontenc}
\usepackage{lmodern}
\usepackage{libertine}

\DefineBibliographyStrings{english}{
  backrefpage  = {},
  backrefpages = {},
}

\DeclareFieldFormat{bracketswithperiod}{\mkbibbrackets{#1}}

\renewbibmacro*{pageref}{%
  \iflistundef{pageref}
    {}
    {\printtext[bracketswithperiod]{%
       \ifnumgreater{\value{pageref}}{1}
         {\bibstring{backrefpages}}
         {\bibstring{backrefpage}}%
       \printlist[pageref][-\value{listtotal}]{pageref}}%
     }
 }
     
\renewbibmacro{in:}{}
\AtEveryBibitem{\clearfield{pages}}

\usepackage[letterpaper,margin=1in]{geometry}
\usepackage{lipsum}

\usepackage{url}            %
\usepackage{booktabs}       %
\usepackage{nicefrac}       %
\usepackage{microtype}      %
\usepackage{enumitem}
\usepackage{dsfont}
\usepackage{multicol}
\usepackage{xcolor}
\usepackage{mdframed}

\usepackage{amsthm,amsfonts,amsmath,amssymb,epsfig,color,float,graphicx,verbatim, enumitem}

\usepackage{algpseudocode,algorithm,algorithmicx}  

\makeatletter
\newcommand{\algmargin}{\the\ALG@thistlm}   
\makeatother
\algnewcommand{\parState}[1]{\State%
    \parbox[t]{\dimexpr\linewidth-\algmargin}{\strut #1\strut}}

\usepackage{mathtools}
\usepackage{bbm}

\usepackage{soul}

\newlist{myEnumerate}{enumerate}{9}
\setlist[myEnumerate,1]{label=(\arabic*)}
\setlist[myEnumerate,2]{label=(\Roman*)}
\setlist[myEnumerate,3]{label=(\Alph*)}
\setlist[myEnumerate,4]{label=(\roman*)}
\setlist[myEnumerate,5]{label=(\alph*)}
\setlist[myEnumerate,6]{label=(\arabic*)}
\setlist[myEnumerate,7]{label=(\Roman*)}
\setlist[myEnumerate,8]{label=(\Alph*)}
\setlist[myEnumerate,9]{label=(\roman*)}

\usepackage{thm-restate}

\definecolor{green}{rgb}{0.0, 0.5, 0.0}
\usepackage[colorlinks,citecolor=blue,linkcolor=magenta,bookmarks=true,hypertexnames=false]{hyperref}
\usepackage[nameinlink,capitalize]{cleveref}

\crefformat{equation}{(#2#1#3)}
\crefname{lemma}{lemma}{lemmata}
\crefname{claim}{claim}{claims}
\crefname{theorem}{theorem}{theorems}
\crefname{proposition}{proposition}{propositions}
\crefname{corollary}{corollary}{corollaries}
\crefname{claim}{claim}{claims}
\crefname{remark}{remark}{remarks}
\crefname{definition}{definition}{definitions}
\crefname{fact}{fact}{facts}
\crefname{question}{question}{questions}
\crefname{condition}{condition}{conditions}
\crefname{algorithm}{algorithm}{algorithms}
\crefname{assumption}{assumption}{assumptions}
\crefname{notation}{notation}{notation}
\crefname{cond}{Condition}{Conditions}
\crefname{contModel}{Contamination Model}{Contamination Models}

  {%
   \par\noindent{\bfseries\upshape Proof Sketch\ }%
  }%
  {\qed}

\newtheorem{theorem}{Theorem}[section]
\newtheorem{lemma}[theorem]{Lemma}

\newtheorem{corollary}[theorem]{Corollary}
\newtheorem{claim}[theorem]{Claim}

\newtheorem{definition}[theorem]{Definition}

\newtheorem{fact}[theorem]{Fact}

\theoremstyle{definition}

\newtheorem{remark}[theorem]{Remark}

\usepackage{fontawesome}

\newlist{itemizec}{itemize}{2}
\setlist[itemizec,1]{label=\faCaretRight ,wide, parsep= 0.05pt, left = 15pt}

\newcommand{\eps}{\epsilon}
\renewcommand{\tilde}{\widetilde}

\newcommand{\Ind}{\mathds{1}}
\newcommand{\1}{\Ind}

\renewcommand{\Pr}{\operatorname*{\mathbf{Pr}}}

\newcommand{\Cov}{\operatorname*{\mathrm{Cov}}}
\newcommand{\E}{\operatorname*{\mathbf{E}}}

\newcommand{\poly}{\operatorname*{\mathrm{poly}}}
\newcommand{\polylog}{\operatorname*{\mathrm{polylog}}}

\newcommand{\Bernoulli}{\mathrm{Ber}}

\renewcommand{\vec}[1]{\boldsymbol{\mathbf{#1}}}

\renewcommand{\d}{\mathrm{d}}

\def\R{\mathbb R}

\def\N{\mathbb N}

\def\Z{\mathbb Z}

\newcommand{\cA}{\mathcal{A}}

\newcommand{\cD}{\mathcal{D}}
\newcommand{\cE}{\mathcal{E}}

\newcommand{\cI}{\mathcal{I}}

\newcommand{\cN}{\mathcal{N}}

\newcommand{\cS}{\mathcal{S}}

\newcommand{\cU}{\mathcal{U}}
\newcommand{\cV}{\mathcal{V}}

\newcommand{\bA}{\vec{A}}

\newcommand{\bI}{\vec{I}}

\newcommand{\hide}[1]{}

\DeclareMathOperator*{\pr}{\mathbf{Pr}}

\newcommand{\normal}{\mathcal{N}}
\def\d{\mathrm{d}}

\newcommand{\tr}{\mathrm{tr}}
\let\vec\mathbf

\def\colorful{0}

\ifnum\colorful=1

\newcommand{\snote}[1]{\footnote{{\bf \color{violet}[Sihan: { #1}\bf ]}}}
\newcommand{\tnote}[1]{\footnote{{\bf [Thanasis: {#1}\bf ] }}}

\newcommand{\inote}[1]{\footnote{{\bf [[Ilias: {#1}\bf ]] }}}

\else

\newcommand{\anote}[1]{}
\newcommand{\tnote}[1]{}
\newcommand{\snote}[1]{}
\newcommand{\inote}[1]{}

\fi

\allowdisplaybreaks

\input{macro}

\title{Entangled Mean Estimation in High-Dimensions}
\author{
Ilias Diakonikolas\thanks{Supported by NSF Medium Award CCF-2107079 and an H.I. Romnes Faculty Fellowship.}\\
University of Wisconsin-Madison\\
{\tt ilias@cs.wisc.edu}\\
\and
Daniel M. Kane\thanks{Supported by NSF Medium Award CCF-2107547 and NSF Award CCF-1553288 (CAREER).}\\
University of California, San Diego\\
{\tt dakane@cs.ucsd.edu}
\and
Sihan Liu\thanks{Supported by NSF Medium Award CCF-2107547 and NSF Award CCF-1553288 (CAREER). Part of this work was done when the author was visiting the National Institute of Informatics (NII).}\\
University of California, San Diego\\
{\tt sil046@ucsd.edu}\\
\and
Thanasis Pittas\thanks{Supported by NSF Medium Award CCF-2107079 and NSF Award DMS-2023239 (TRIPODS).}\\
University of Wisconsin-Madison\\
{\tt pittas@wisc.edu}\\
}

\begin{document}

\maketitle

\begin{abstract}
We study the task of high-dimensional entangled mean estimation in the 
subset-of-signals model. Specifically, given $N$ independent random points 
$x_1,\ldots,x_N$ in $\R^D$ and a parameter $\alpha \in (0, 1)$ such that 
each $x_i$ is drawn from a Gaussian with mean $\mu$ and unknown covariance, and an unknown $\alpha$-fraction of the points have {\em identity-bounded} 
covariances, the goal is to estimate the common mean $\mu$. The one-dimensional 
version of this task has received significant attention in theoretical computer 
science and statistics over the past decades. Recent work~\cite{LiaYua20, 
compton2024near} has given near-optimal upper and lower bounds for the one-dimensional setting. On the other hand, our understanding of even 
the information-theoretic aspects of the
multivariate setting has remained limited.

In this work, we design a computationally efficient algorithm achieving an information-theoretically near-optimal error. Specifically, we show that the optimal 
error (up to polylogarithmic factors) is $f(\alpha,N) + \sqrt{D/(\alpha N)}$, 
where the term $f(\alpha,N)$ 
is the error of the one-dimensional problem and the second term 
is the sub-Gaussian error rate.
Our algorithmic approach employs an iterative refinement 
strategy, whereby we progressively learn more accurate 
approximations $\hat \mu$ to $\mu$. This is achieved via a 
novel rejection sampling procedure that removes points  
significantly deviating from $\hat \mu$, 
as an attempt to filter out 
unusually noisy samples. A complication that arises is that 
rejection sampling introduces bias in the distribution of the 
remaining points. To address this issue, 
we perform a careful analysis 
of the bias, develop an iterative dimension-reduction strategy, 
and employ a novel subroutine inspired by list-decodable learning 
that leverages the one-dimensional result.
\end{abstract}

\thispagestyle{empty}

\newpage
\setcounter{page}{1}

\section{Introduction}\label{sec:intro}

Classical statistics has traditionally focused on the idealized scenario where 
the input dataset consists of independent and identically distributed 
samples drawn from a fixed but unknown distribution. 
In a wide range of modern data analysis applications, there is an increasing need to move beyond this assumption since datasets are often collected from heterogeneous sources \cite{dundar2007learning,steinwart2009fast,zhu2014correlated,fan2014challenges,flaxman2015gaussian,gao2022survey}.
A natural formalization of heterogeneity in the context of mean estimation (the focus of this work) 
involves having each datapoint drawn independently
from a potentially different distribution within a (known) family that shares a \emph{common} mean parameter. Distributions with this property are referred to as \emph{entangled}, and the setting is also known as sample \emph{heterogeneity} or \emph{heteroskedasticity}. 

The task of estimating the mean of entangled distributions has gained significant 
attention in recent years for a number of reasons.
First, from a practical viewpoint, entangled distributions intuitively 
capture the idea of collecting samples from diverse sources. One of the early works 
that studied this task~\cite{ChiDKL14} illustrates this with the following 
crowdsourcing example. Suppose that multiple users rate a product with 
some true value $\mu$. Each user $i$ has their own level of knowledge about the 
product, captured by a standard deviation parameter $\sigma_i$. The rating from user 
$i$ is {assumed to be} sampled from a Gaussian distribution 
with mean $\mu$ and covariance $\sigma_i^2$, 
and the goal is to estimate $\mu$ in small absolute error using these samples.
Other practical examples include datasets collected from sensors under varying 
environmental conditions; see, e.g.,~{\cite{kamm2023survey}}.

From a theoretical viewpoint, statistical estimation given access to 
non-identically distributed, heterogeneous data is a natural 
and fundamental task, whose roots trace back several decades 
in the statistics literature. Early work~\cite{hoeffding1956distribution,sen1968asymptotic,weiss1969asymptotic,
sen1970note,stigler1976effect,shorack2009empirical} studied the asymptotic properties 
of such distributions. Specifically, \cite{ibragimov1976local,beran1982robust} 
studied maximum likelihood estimators and \cite{nevzorov1984rate,hallin1997unimodality,hallin2001sample,mizera1998necessary} analyzed the median estimator for non-identically distributed samples. 
Heterogeneity has also been studied for moments of distributions \cite{gordon2006minimum} and linear regression \cite{el1999l1,knight1999asymptotics}.
Mean estimation for entangled distributions, including 
the Gaussian setting considered in~\cite{ChiDKL14}, 
is also related to the classical task of parameter 
learning for mixture models---albeit in a 
regime that is qualitatively different than the one commonly studied. 
While in the canonical setting---see~\cite{Das99,AroKan01,AchMcs05,KanSV05} for classic references and~\cite{belkin2015polynomial,moitra2010settling,ChaSV17,HopLi18,KotSS18,DiaKS18-list,kong2020meta,DiaKKLT22-cluster,BakDJKKV22,LiuLi22,DiaKKLT22-cluster,diakonikolas2023sq}
for more recent work---one typically assumes a small (constant) 
number of components $k \ll N$ with different means, 
in the entangled setting each sample comes from its own component ($k=N$). Importantly, the shared mean assumption 
allows for meaningful results 
despite the high number of components.

\smallskip

\noindent {\bf Prior Work} 
We now summarize prior work for mean estimation of entangled Gaussians, 
starting with the (now well-understood) one-dimensional case. 
In this setting, we have access to samples $x_i \sim \cN(\mu,\sigma_i^2)$ 
with unknown $\sigma_i$ values.
For a concrete and simple configuration for the $\sigma_i$'s, 
we  consider the so-called \emph{subset-of-signals model}, 
introduced in \cite{LiaYua20}. In this model, it is assumed that 
at least an $\alpha$-fraction of the samples have $\sigma_i \leq 1$, 
while the remaining can have arbitrary variances. 
The goal is to estimate $\mu$ in absolute error that 
is as small as possible in terms of the number of samples $N$ and the rate $\alpha$.
A series of works \cite{ChiDKL14,PenJL19,PenJL19-isit,pensia_estimating_2021,xia2019non,yuan2020learning,LiaYua20,DevLLZ23,compton2024near} has established upper and lower bounds for this task. Specifically, the recent 
work \cite{compton2024near} gave an estimator 
with error matching (up to polylogarithmic factors) 
the lower bound of \cite{LiaYua20} in the subset-of-signals model 
(for a very wide regime of $\alpha$ values). Entangled Gaussian mean estimation 
in the subset-of-signals model is thus essentially resolved in one dimension.
Additional discussion on related work is provided in \Cref{app:additional_related_work}.

\smallskip

\noindent {\bf Entangled Mean Estimation in High Dimensions}
In contrast, the multivariate version of this problem is much less understood. 
Some of the prior work \cite{ChiDKL14,PenJL19,pensia_estimating_2021,compton2024near} 
only tangentially considered higher dimensions,  
focusing on the rather restricted setting that the covariance matrices are 
spherical, i.e., of the form $\vec \Sigma_i = \sigma_i^2 \vec I$. 
For this specific special case, it turns out that the problem becomes {\em easier} 
in higher dimensions---as each coordinate provides {more} 
information about {the scalar parameter $\sigma_i$.}
A more general formulation would 
be to replace the sphericity assumption on the $\vec \Sigma_i$'s by a boundedness assumption.
This leads to the following high-dimensional 
formalization of the subset-of-signals model. %

\begin{restatable}[Subset-of-Signals Model For High-Dimensional Gaussians]{definition}{ORIGINALMODEL}\label{def:model}
 Let $\mu \in \R^D$ be a target vector and $\alpha \in (0,1)$ be a parameter. A set of $N$ points in $\R^D$ is generated as follows: First, an adversary chooses $N$ positive semidefinite (PSD)  matrices $\vec \Sigma_1,\ldots, \vec \Sigma_N \in \R^{D \times D}$ under the 
 constraint that $\sum_{i=1}^N\1( \vec \Sigma_i \preceq \vec I) \geq \alpha N$. 
 Then, for each $i=1,\ldots,N$, the sample $x_i$ is drawn independently from $\cN(\mu,\vec \Sigma_i)$. The final dataset $\{x_1,\ldots,x_N\}$ is the input 
 provided to the learning algorithm. We call $\mean$ the \emph{common mean} and $\alpha$ the \emph{signal-to-noise rate} of the model.
\end{restatable}

\noindent This natural definition was 
suggested by Jerry Li~\cite{Li24-TTIC} at the TTIC 
Workshop on New Frontiers in Robust Statistics, 
where the complexity of the problem was posed as 
an open question.

We emphasize that our understanding of entangled mean estimation 
in the aforementioned setting is fairly limited---even information-theoretically.
The results in \cite{ChiDKL14, compton2024near, LiaYua20} already imply 
that any estimator for the arbitrary covariance setting must incur 
error that is larger, by at least a polynomial factor, 
than the error achievable in the spherical covariance case 
(see \Cref{sec:comparison} for more details). This suggests that the 
bounded covariance setting is more challenging than the spherical case 
and requires new ideas. Specifically, prior to this work, the optimal 
rate for the bounded covariance case was open---even ignoring 
computational considerations. 

A standard attempt to obtain a (potentially tight) 
upper bound on the error involves
using the one-dimensional estimator along an exponentially large cover 
of the unit ball in $\R^D$, and combining these estimates into a vector via a linear program (ala Tukey median)~\cite{Tuk75}. Unfortunately, this approach may fail in our setting, 
due to the following issue. 
Establishing correctness of the approach 
requires that 
the failure probability of the one-dimensional 
estimator is exponentially small in $D$. 
However, the currently best known error 
guarantees~\cite{compton2024near} 
{hold only with probability 
$1 - \poly(N)$.}\footnote{
{
Though one could amplify the success probability of \cite{compton2024near}
in a black-box manner using the standard ``median trick'', we remark that such a strategy will lead to a factor of $D$ loss in the error guarantee if the goal is to achieve success probability $1 - \exp( -\Theta(D) )$.
}
} Moreover, even if this 
obstacle could be circumvented, we would still 
end up with an exponential-time estimator. Finally, we note that 
a simple and natural computationally efficient approach 
involves applying a one-dimensional estimator 
for each axis of the space. Unfortunately, the 
error incurred by this approach is $\sqrt{D}$ 
times that of the one-dimensional estimator, 
which turns out to be significantly suboptimal.

 In summary, none of the known approaches yields error better than $\poly(D) f(\alpha,N)$, where $f(\alpha,N)$ is the error of the optimal one-dimensional estimator, {leaving even the information-theoretically tight bound wide open.} This leads to the core question of our work:
 \begin{center}
\emph{What is the optimal error rate for high-dimensional entangled mean estimation, 
both \\(i) from an information-theoretic perspective, and (ii) for computationally efficient algorithms?}
\end{center}
In this work, we resolve 
both aspects of this question (up to polylogarithmic factors) 
for a wide range of the parameters $N,D,\alpha$.

\subsection{Main Result}
Before we formally state our contributions, 
we recall the error guarantee of the $1$-d estimator 
given in \cite{compton2024near}. In particular, if
we denote by $N$ the number of samples and $\alpha$ the signal-to-noise rate (fraction of points with variances bounded from above by one), 
then their estimator $\hat \mu \in \R$ satisfies $|\hat \mu - \mu| \leq f(\alpha,N)$ with high probability, where
\begin{align}
\label{eq:simple-f-def}    
    f(\alpha,N) =  (\log(N/\alpha))^{O(1)} \cdot \begin{cases}
          \alpha^{-2} N^{-3/2},  &\displaystyle \Omega \lp( (\log N)/N \rp)  \leq \alpha \leq {N^{-3/4}} \\[0pt]
        \displaystyle {\alpha^{-2/3} N^{-1/2}}, &\displaystyle {N^{-3/4}} < \alpha < 1 \\[0pt]
        \infty , &\text{otherwise} \;.
    \end{cases}
\end{align}
The above error upper bound had been previously shown~\cite{LiaYua20} 
to be best possible up to polylogarithmic factors 
in the regime $\Omega \lp( \log N / N \rp) \leq \alpha \leq O(N^{1-\eps})$ for any arbitrarily small constant $\eps>0$.

Roughly speaking, the error of our high-dimensional estimator is equal, up to polylogarithmic factors, 
to the sum of the above $1$-d error and the statistical error for mean estimation 
of isotropic Gaussians. Specifically, our main result is the following: 
\begin{theorem}[High-Dimensional Entangled Mean Estimation]\label{thm:main}
\Call{EntangledMeanEstimation}{$N$} in \Cref{alg:mean_estimation} satisfies the following guarantee: 
    The algorithm draws $N$ samples in $\R^D$ from the subset-of-signals model of \Cref{def:model} with common mean $\mu \in \R^D$ and signal-to-noise rate $\alpha \in (0,1)$. If $N \geq \tfrac{D}{\alpha} \log^C(\tfrac{D}{\alpha})$, where $C$ is a  sufficiently large absolute constant, the output $\hat{\mu} \in \R^D$ of the algorithm satisfies the following with probability at least $0.99$:
    \begin{align*}
        \| \hat \mu - \mu \|_2 & \leq \log^{O(1)}(N) \left(\sqrt{\frac{D}{\alpha N} } + f(\alpha,N)  \right) \;,
    \end{align*}
    where $f(\cdot)$ is the function defined in \Cref{eq:simple-f-def}. Moreover, the algorithm runs in time $\poly(D,N)$.
\end{theorem}

We remark that the error bound achieved by our algorithm 
is optimal up to poly-logarithmic factors 
in the subset-of-signals model, {provided that $N \geq \tilde \Omega(D/\alpha)$.} 
To show that the second error term, $f(\alpha, N)$, is necessary, we can simply embed the $1$-d hard instance of \cite{LiaYua20} in the $D$-dimensional space. 
Specifically, we can set the mean and variance of the first coordinate according 
to the $1$-d hard instance, and set the remaining coordinates to be deterministically $0$. 
The second term $\sqrt{D / (\alpha N)}$ is the statistical error rate of estimating 
the mean of isotropic Gaussians. This term is also necessary, as 
can be seen by embedding the standard hard instance of $D$-dimensional 
isotropic Gaussian mean estimation into the $\alpha N$ many samples with bounded covariances, and setting the covariances of the rest of the samples to be 
sufficiently large so that  they reveal almost no information.\looseness=-1

Finally, the algorithm succeeds whenever $N > D/\alpha$ (times polylogarithmic factors).
We remark that this is necessary for any estimator to achieve errors 
smaller than a constant, i.e., $\eps < 1/2$, even when the 
identities of the samples with bounded covariances are revealed to the 
algorithm. Extending the result to any $N \in \Z_+$ is an interesting 
open question that we leave for future work.

\subsection{Brief Overview of Techniques}

In this section, we summarize our approach for obtaining an estimator achieving the guarantees of \Cref{thm:main}. Towards this end, we will start by explaining how to obtain an initial rough estimate 
$\tilde \mu$ such that $\|\tilde \mu - \mu\|_2 \lesssim \sqrt{D}$.  We note that the main novelty (and bulk) 
of our technical work will be on developing 
a recursive procedure that iteratively improves upon $\tilde \mu$.

We now provide an efficient method to achieve the warm start. Specifically, provided that we are in the regime where $f(\alpha, N) = O(1)$, 
such an estimate $\tilde \mu$ can be easily obtained by running 
the $1$-d estimator from \cite{compton2024near} along each axis.
For the other regime, we design a sophisticated tournament 
procedure that outputs an estimate within $O(\sqrt{D})$ from the true mean (cf. \Cref{sec:tournament}). 

We next describe how to achieve improved estimation accuracy. 
Na\"ive approaches such as sample means are destined to fail in the 
subset-of-signals model, due to the fact that no assumptions are made 
on the covariances of $(1-\alpha)$-fraction of the samples. 
Specifically, these matrices can have arbitrarily large operator norms, 
which can cause the average of the samples to suffer from arbitrarily 
high {statistical errors in $\ell_2$ distance}. Our approach to 
overcome this issue 
is to use our initial estimator $\tilde \mu$ (warm start) 
to detect and reject samples that are too far from the true mean $\mu$ in Euclidean norm. 
Algorithmically, the rejection sampling procedure is to accept each sample $x$ with probability 
$\exp( - \|  x - \tilde \mu \|_2^2 / D )$. 
On the one hand, most of the ``good'' samples (the $\alpha N$ points 
with bounded covariance) survive the rejection sampling: this is because a Gaussian sample with covariance bounded above by $\vec I$ is $O(\sqrt{D})$-far from $\mu$ (and hence from $\tilde \mu$) with high probability, which causes its acceptance with high probability.
Regarding the remaining points (i.e., the ones 
with covariances $\vec \Sigma_i$ such that 
$ \tr\lp( \vec \Sigma_i \rp) \gg D$), 
the rejection sampling step %
ensures that the probability of acceptance is small enough so that
the average of the covariance matrices of the surviving points (inliers and outliers), which we denote by $\Sigmavg$, 
will have its trace bounded from above by $O(D)$.
We show that this essentially allows for accurate estimation of 
the population mean of the surviving points; roughly speaking, 
this follows from the fact that the standard error for mean estimation 
of a bounded covariance random variable $X$ is $O(\sqrt{\tr \lp( \Cov(X) \rp)}/N)$.

Unfortunately, the aforementioned approach does not quite work 
for the following reason: 
the rejection sampling procedure will also cause the population mean 
of the surviving samples to be biased.
In particular, since the acceptance probability is given by the exponential of some quadratic in the input point, 
the resulting distribution of each point $x_i$ conditioned on its acceptance 
will be some new Gaussian distribution 
whose means and covariances are functions of the true mean $\mu$, 
the covariance $\vec \Sigma_i$ of $x_i$, and the center $\tilde \mu$ used in the rejection sampling.
After a careful calculation, one can show that the bias 
of the new population mean of a set of samples,  
conditioned on their acceptance, will be given by 
$ \Sigmavg \; \lp(  \mu - \tilde \mu \rp)/D $, 
where $\Sigmavg$ is the average of the conditional 
covariance of the surviving points.
Since the operator norm of $\Sigmavg$ could be as large as 
its trace, which could itself be as large as $D$,
the bias caused by the rejection sampling over 
the \emph{entire space} could hence be prohibitively large.
That being said, if we can find some \emph{low-variance} subspace 
$\cV$ such that $v^\top \Sigmavg v$ is small for any $v \in \cV$, 
the magnitude of the bias %
within $\cV$ will be only a small constant multiple 
of $\| \mu - \tilde \mu \|_2$.
Fortunately, since the trace of $\Sigmavg$ is at most $O(D)$, 
by a simple averaging argument, 
$\Sigmavg$ must have at least $D/2$ many eigenvalues 
that are at most $O(1)$. Thus, the subspace spanned 
by the corresponding eigenvectors 
gives the desired \emph{low-variance} subspace.
Moreover, we show that this subspace can be approximately computed from the samples 
(more precisely, we can compute a subspace of \emph{similarly} 
low variance, up to polylogarithmic factors).
It then remains to estimate the mean $\mu$ 
within the complement high-variance subspace $\cV^{\perp}$.
To achieve this goal, the idea is to just project 
the datapoints onto $\cV^{\perp}$, and recursively run 
the same algorithm in that lower dimensional 
(of dimension $D/2$) subspace.\footnote{Here we assume 
that the algorithm can take multiple datasets, 
where each of them is generated by \Cref{def:model}.
We argue that this can be easily simulated with a single dataset  generated by \Cref{def:model}; see \Cref{def:model2} for more details.}
When the recursive procedure reaches a subspace with dimension $\polylog(D)$, we can simply run the $1$-d estimator 
along each axis to finish the recursion.
Since $\cV^{\perp}$ now has dimension only $D/2$, 
the recursion terminates after $\log D$ many iterations.

In our description so far, one full execution of the recursive algorithm yields some $\hat \mu$ with error $ \| \hat \mu - \mu  \|_2 \leq \|  \tilde \mu -  \mu  \|_2 / 2 + \sqrt{D/(\alpha N)} + f(\alpha, N)$ (up to polylogarithmic factors). 
The first term corresponds to the bias caused by the rejection sampling and 
the second term corresponds to the statistical error of $D$-dimensional mean estimation 
of the ``good'' $\alpha N$ samples of bounded covariance. Finally, 
the last term corresponds to the error of the $1$-d estimator used in the base case of the recursion.
Thus, the execution of this recursive algorithm improves the estimation error 
by a constant factor, provided that $\|  \tilde \mu -  \mu  \|_2$ 
is still significantly larger than the error bound stated in \Cref{thm:main}.
By iteratively repeating this process for $\log(ND)$ many iterations, 
the estimation error can be brought down {to the error bound of  \Cref{thm:main}.}

\subsection{Related Work}\label{app:additional_related_work}

\subsubsection{Additional Related Work on Entangled Mean Estimation}

{In this section we discuss the works that are most closely related to 
this paper. We refer the reader to \cite{PenJL19,PenJL19-isit,pensia_estimating_2021} and \cite{DevLLZ23} for additional 
references and in-depth discussion of earlier work 
in the statistics literature. 

The work of \cite{ChiDKL14} studied entangled mean estimation 
in one dimension. Instead of assuming that a subset of the samples have 
bounded variances, like in \Cref{def:model}, 
the $N$ samples are $x_i \sim \cN(\mu,\sigma_i^2)$ 
with $\sigma_1 \leq \sigma_2 \leq \cdots \leq \sigma_N$ 
and the error guarantee is stated directly as a function of 
the $\sigma_i$'s. 
{They show that the best possible estimation error is on the order of $ 1 / \sqrt{ \sum_{i=1}^N \sigma_i^{-2} }$ when the variances are known a priori.
Otherwise, they show an error of $\tilde O( \sqrt{N} \sigma_{\log n} )$ is achievable in the absence of such knowledge.}
\cite{ChiDKL14} also studies the high-dimensional setting where the samples follow spherical Gaussian distributions, i.e., with covariances equal to $\sigma_i^2 \vec I$.
As already mentioned, the task with the spherical covariances 
becomes easier in higher dimensions, 
as it is possible to estimate the covariance scale parameter $\sigma_i$. Using this, they achieve an error bound 
on the order of 
$1 / \sqrt{ \sum_{i=2}^N \sigma_i^{-2} }$ in $\ell_\infty$ distance when $D \gg \log N$. Notably, this almost recovers the error bound 
when the covariances are known a priori, except for missing the 
dependency on $\sigma_1$.

Subsequent works \cite{PenJL19,PenJL19-isit,pensia_estimating_2021} 
explore the more challenging non-Gaussian setting under only the 
assumptions of unimodality and radial (spherical) symmetry.
Specifically, they give a hybrid estimator
that achieves an error rate of {$ \sqrt{D}   \sqrt{N}^{1/D} \sigma $} provided that at least $\Omega \lp( \log N / N \rp)$-fraction of samples have marginal variance bound $\sigma$. 
One could see that this result recovers that of \cite{ChiDKL14} by setting $D = 1$.
The authors also consider the more general settings where the distributions are only assumed to be \emph{centrally} symmetric, i.e., 
the density function $\rho: \R^d \mapsto \R_+$ satisfies $\rho(x) = \rho(-x)$, and achieve an error of $O(\sqrt{N})$. This setting covers our setup of non-spherical Gaussians. 
Yet, as pointed out in \cite{compton2024near}, the error bound given in \cite{PenJL19} is sub-optimal under the subset-of-signals model even in the $1$-d case.

The work of \cite{DevLLZ23} also uses
symmetry in place of Gaussianity.
Their algorithm is fully adaptive, i.e., requiring no parameter tuning for specific distribution families, and is made possible by the techniques of intersecting confidence intervals, which {has later inspired the work of \cite{compton2024near}} that leads to (nearly) optimal $1$-d estimators in the subset-of-signals model.

The work \cite{yuan2020learning,LiaYua20} introduced the subset-of-
signals model, provided a nearly optimal lower bound within the model, 
and showed theoretical guarantees for \emph{the iterative trimming 
algorithm}---a widely used heuristic for entangled mean estimation. 
Notably, the algorithm works by iteratively searching for a mean 
parameter that minimizes the square distance to a subset of samples, 
and then searching for a subset of samples that minimize the squared distance to a 
given estimate of the mean parameter. Our algorithm bears some 
similarity to these techniques as we are also using estimates from past 
iterations to perform rejection sampling on samples, and then use the 
surviving samples to construct new estimates. We refer the readers to 
\Cref{sec:mean_estimation_main} for a detailed outline of our 
techniques.

Finally, Theorem 1 in \cite{compton2024near} (restated as \Cref{thm:one-dim} here) gives a nearly optimal 
$1$-d estimator for the subset-of-signals model. Similar to \cite{ChiDKL14}, they show that the result can be easily applied in the multivariate spherical Gaussians setting to nearly recover the error bound achievable with prior knowledge on variance scales. Moreover, as an improvement, they only require the dimension to be at least $2$ for the multivariate bound to be effective. 

\subsubsection{Comparison of Optimal Error in Spherical vs Arbitrary Gaussians} \label{sec:comparison}
From the results of \cite{ChiDKL14,compton2024near} and \cite{LiaYua20}, it can be seen that 
there is a polynomial gap 
between the 
the errors in the cases of spherical Gaussians and those with arbitrary covariances {in the subset-of-signals model}.
The first observation is that in the arbitrary covariance matrix setting, any estimator must have an error of $f(\alpha, N)$ (up to polylogarithmic factors) as shown in \cite{LiaYua20} (this is because by allowing arbitrary covariances one can encode the hard one-dimensional instance in one of the axes). 
{This means that the error is at least $f(\alpha,N) \geq \tilde \Omega (1/(\alpha^{2/3}\sqrt{N}))$.}
On the other hand, for \emph{spherical} Gaussians, \cite{ChiDKL14} and \cite{compton2024near} give estimators with errors $O( 1 / \sqrt{ \sum_{i=1}^N \sigma_i^{-2} })$. {Using $\sigma_1\leq\ldots \leq \sigma_{\alpha N}\leq 1$ to translate that to the subset-of-signals model, the estimator for the spherical Gaussians can achieve an error $O(1/\sqrt{\alpha N})$. This shows that there is a $\poly(1/\alpha)$ gap between the optimal errors in the two settings.}

\subsubsection{Further Related Work}
\paragraph{Robust Statistics}
From a robustness perspective,  samples with arbitrarily large 
covariances in \Cref{def:model} can be viewed as outliers. Robust 
statistics typically considers stronger outlier models, 
such as the \emph{Huber contamination model} \cite{Hub64}, 
where each outlier point is sampled from an unknown and potentially 
arbitrary distribution; or the \emph{strong contamination model} 
where outliers can be adversarial and even violate 
the independence between samples. Unlike the results of 
this paper, consistency in these more challenging models 
is impossible, with the error {bounded from below} 
by some positive function of the fraction of outliers. 
Efficient high-dimensional mean estimation in the aforementioned corruption models was first achieved in 
\cite{DiaKKLMS16-focs,LaiRV16}, where the outlier fraction is a 
constant smaller than $1/2$. When this fraction is more than $1/2$, it 
is no longer possible to produce a single estimate with worst-case 
guarantee. This setting is known as ‘list-decodable’ mean estimation, 
first studied in \cite{ChaSV17}. We refer to the recent book 
\cite{DiaKan22-book} for an overview of algorithmic robust statistics.

\paragraph{Other models of semi-oblivious adversaries}
Similar to the subset-of-signals model, there are other frameworks for modeling less adversarial outliers that are assumed to be independent and have additional (Gaussian) structure. These models are commonly referred as semi-oblivious noise models. One such model assumes that inlier points are distributed as $x_i \sim \cN(\mu,1)$ and an $\alpha<1/2$ fraction of outlier points are sampled as $x_i \sim \cN(z_i,1)$ where $z_i$'s are arbitrary centers. 
Recent work \cite{kotekal2024optimal} has characterized the error for estimating $\mu$ in this model, improving upon a line of work \cite{collier2019multidimensional,cai2010optimal,jin2008proportion}. \cite{collier2019multidimensional} studies the multivariate extension of the model. Finally, beyond mean estimation, the idea of modeling outliers via mean shifts has also been explored in the context of regression in \cite{sardy2001robust,gannaz2007robust,mccann2007robust,she2011outlier}.

\section{The Dimensionality Reduction Algorithm and Proof Roadmap}\label{sec:mean_estimation_main}

This section describes our main algorithm in tandem with a detailed sketch of its correctness proof. 
The pseudocode is provided in \Cref{alg:mean_estimation,alg:recursive,alg:find-subspace,alg:partial-estimate}.
{\Cref{alg:mean_estimation} describes our main algorithm that computes a rough initial estimate, and then leverages a dimensionality reduction routine to iteratively improve the estimation error.} 
\Cref{alg:recursive} contains the main subroutine~\recursivefunction.
Aided by two subroutines,~\SearchSubSpace~and~\PartialEstimate~(cf. \Cref{alg:find-subspace,alg:partial-estimate}) whose functionalities will be explained later on,
~\recursivefunction~takes as input an estimate $\hat \mu \in \R^D$ and 
a batch of $\Theta \lp( N / \big( \log N \log D \big)\rp)$ independent samples from 
the data-generating model of \Cref{def:model}, 
and produces a more refined estimate $\hat \mu'$ satisfying the following 
guarantee:  if $\| \hat \mu - \mu \|_2$ is significantly larger 
than the error bound specified in \Cref{thm:main}, then the new estimate 
$\hat \mu'$
satisfies $ \| \hat \mu' - \mu \|_2 \leq \| \hat \mu - \mu \|_2 / 2$.
This is summarized in the following statement (see \Cref{lem:single_stage_analysis} for the formal statement). 
\footnote{As will be explained later on, the routine is designed to operate in any subspace $\cV$ provided in the input. For this reason, the formal statement in \Cref{lem:single_stage_analysis} uses a matrix $\RowP$ whose rows are the orthonormal vectors that span that subspace ($\RowP^\top \RowP$ is the orthogonal projection matrix). What we present here corresponds to the simplified case $\RowP=\vec I$ and $d=D$.\looseness=-1}

\begin{algorithm}[h]
\caption{Entangled Mean Estimation in High Dimension}
\label{alg:mean_estimation}
\begin{algorithmic}[1]

\Function{EntangledMeanEstimation}{$N$}
    \State \textbf{Input}: Total number of samples $N$ to use, and noise-to-signal ratio $\alpha \in (0,1)$.
    \State \textbf{Output}: $\mu' \in \R^D$.
    \State $m \gets \mm$, $r \gets \rr$ \Comment{max 
    depth of recursion and number of outer loop iterations}
    \State {Set $\delta$ to be some sufficiently large constant. \footnotemark}
    \State $\tau \gets N^{-\delta}/r,n \gets \tfrac{N}{2+ m ( 3r + 1)}$ \Comment{Probability of failure $\tau$ and sample budget $n$ per iteration.}
    \State $\hat{\mu} \gets \Call{TournamentImprove}{\vec 0,n,\tau}$.\label{line:initial_guess}  \Comment{Rough estimate $\hat{\mu}$ with $O(\sqrt{D})$ error (\Cref{lem:guess-center})}
    \State \textbf{If } $ \f(\alpha, n) \geq \sqrt{D} $ \textbf{ then return } $\hat \mu$ \label{line:just-return}
    \Comment{where $\f(\cdot)$ is defined in \Cref{eq:function_f}}

    \For{$i=1,\ldots,r$}\label{line:stages} \Comment{each iteration improves estimation error}
    \State $\hat \mu' \gets$  \Call{RecursiveEstimate}{$\vec I,n,\hat \mu,\tau$}, 
    \label{line:root} \Comment{Iterative improvement (cf. \Cref{alg:recursive})}
    \State{$\hat{\mu} \gets \hat \mu'$}
    \EndFor
    \State \textbf{return} $\hat{\mu}$ \label{line:return}

\EndFunction

\end{algorithmic}
\end{algorithm}
\footnotetext{More formally, with $n$ input samples, the $1$-d estimator from \cite{compton2024near} with probability $1 - n^{-\delta}$ achieves error $\f(\alpha, n)$ (cf. \Cref{eq:function_f}).
For the purpose of our algorithm, we require its success probability to be a sufficiently large polynomial in $n$. Hence, we set $\delta$ to be some sufficiently large constant.}

\vspace{-1em}
\begin{lemma}[Iterative Refinement (Informal; see \Cref{lem:single_stage_analysis})]\label{lem:single_stage_analysis_informal}
There exists an algorithm \textsc{RecursiveEstimate} that takes as input $\hat \mu \in \R^D$ satisfying $ \| \hat \mu - \mu \|_2 \lesssim \sqrt{D}$, uses a dataset of $n=\Theta \lp( N / \big( \log N \log D \big)\rp)$ independent samples from the model of \Cref{def:model}, and produces some $\hat \mu'$ such that the following holds with high probability (where $f(\alpha,n)$ is defined in \Cref{eq:simple-f-def}):
    \begin{align*}
        \left\|  \hat \mu' -  \mu \right\|_2 \leq 
        \frac{1}{2}\left\|  \hat \mu -  \mu \right\|_2  + \polylog(ND) \lp( \sqrt{\frac{ 
D}{\alpha N} }  + f(\alpha,n)\rp) \;.
    \end{align*}
\end{lemma}
Suppose that we have a rough estimate $\hat \mu$ satisfying 
$\| \hat \mu - \mu \|_2 \lesssim \sqrt{D}$.
It is easy to see that running~\recursivefunction~for at most $\Theta(\log(N))$ many iterations\footnote{For the purpose of this introductory section, we will pretend as if the algorithm can draw an independent dataset from \Cref{def:model} in each iteration. This type of access can be simulated by randomly splitting a large dataset into smaller ones and is presented formally in \Cref{sec:preliminaries}. \label{fnote:samples}} will bring this error down to the one presented in \Cref{thm:main}.
This iterative reduction is accomplished using the for loop in \Cref{alg:mean_estimation}.

We now summarize the idea behind \Cref{lem:single_stage_analysis_informal}. 
Let $\hat \mu$ be the rough estimate provided as input. 
The routine~\recursivefunction~is a recursive function that performs 
divide-and-concur on the dimension $D$.
In particular, in each step, it receives a batch of $n=\Theta \lp( N / (\log (N) \log (D))\rp)$ many samples, and uses the subroutine~\SearchSubSpace
to partition the current subspace further into a low-variance 
and a high-variance subspace\footnote{It will be explained later 
on in this introductory section in what sense we call the subspace ``low-variance''.}, each with dimension at most half 
of the original subspace.
Then the algorithm invokes~\PartialEstimate~to compute an estimate within the low-variance subspace that has improved error; while for the high-variance subspace it makes a recursive call.
The recursion continues until we reach a subspace of dimension $d \lesssim \polylog(ND)$.
In that case, we can simply use the $1$-d estimator from 
\cite{compton2024near} along each axis in that low-dimensional subspace; 
the corresponding error guarantee will be at most $\sqrt{d} = \polylog(ND)$ factor larger than that of the error $f(\alpha,N)$ of the $1$-d base estimator (\Cref{thm:one-dim}).

Suppose that the dimension of the current subspace is $d$, i.e., assume that the samples are now vectors in $\R^d$ {with mean $\mu \in \R^d$}.
The analysis consists of the following inductive claim: as long as the dimension is not too small, i.e., $d \gg \polylog(ND)$, 
a single step of~\recursivefunction~produces a subspace $\cV \subset \R^d$ of dimension $\dim(\mathcal V)\leq d/2$ (this is the low-variance subspace identified by~\SearchSubSpace) together with an estimate $\mul \in \cV$ such that with high probability
\begin{align}\label{eq:goal}
    \|\mul -  \vec \Pi_{\cV}\mu \|_2 \leq \|  \hat{\mu} -\mu \|_2/(2 \log_2 D) + \tilde O \lp( \sqrt{{ 
\dim(\mathcal V)}/{(\alpha N)} } \rp) \, ,
\end{align}
where $\vec \Pi_{\cV}$ is the orthonormal projector onto $\cV$.
If $D$ denotes the original dimension of the space for which \recursivefunction~is first called, after the end of all recursive calls, 
the algorithm will have partitioned the entire space $\R^D$ into at most 
$\log_2 D$ orthogonal low-variance subspaces and will have produced an 
estimate for each of these subspaces with error as in \Cref{eq:goal}. 
The final estimate $\hat \mu'$ returned is a combination of the estimate for each subspace as well as the base case estimator (which has error $\polylog(ND) f(\alpha,N)$).
By summing all the errors together, it can be seen that the error of $\hat \mu'$ is the one provided in \Cref{lem:single_stage_analysis_informal}.
The rest of this section provides a sketch of the correctness proof for \Cref{eq:goal}.\looseness=-1

\subsection{Warm-start Estimate via Tournament}
\recursivefunction~from {\Cref{lem:single_stage_analysis_informal}} requires some $\hat \mu$ such that $ \|\hat \mu - \mu \|_2 \lesssim \sqrt{D}$.
If the $1$-d estimator achieves accuracy 
$\eps = O(1)$, we can simply use it along each axis to obtain such a $\hat \mu$.
In the more challenging regime when $\eps = \omega(1)$, we will exploit the fact that a sample with covariance 
bounded by $\vec I$ is within distance of $\sqrt{D}$ 
from $\mu$ in expectation. 
While accurately identifying such a sample is hard, 
by taking roughly $1 / \alpha$ many samples $\mu_1,\ldots,\mu_{1/\alpha}$, one 
$\mu_i$ of them will satisfy $\|\mu_i - \mu\|_2 \lesssim \sqrt{D}$ with high probability.
We then design a \emph{tournament} procedure to choose a vector $\tilde \mu$ from $\mu_1, \ldots, \mu_{1/\alpha}$ which is almost as good as that $\mu_i$.\looseness=-1

\begin{lemma}[Tournament (Informal; see \Cref{lem:prune})]
\label{lem:prune_informal}
There exists an algorithm that takes as input 
a list $L = \{ \mu_1, \cdots, \mu_k \} \subset \R^D$, 
draws a dataset of size $n$ {from the model of \Cref{def:model2}},
and outputs an estimate $\mu_j \in L$ 
such that
$ \| \mu_{j} - \mu \|_2 \lesssim 
\min_{i \in [k]} \| \mu_i - \mu \|_2 + \f(\alpha, n)$ with high probability.
\end{lemma}
The algorithm establishing the above lemma is inspired by  
ideas {developed in the context of} 
list-decodable mean estimation \cite{DiaKK20-multifilter}.
In particular, for every pair of candidate estimates $\mu_j, \mu_{\ell}$,
the algorithm compares $\mu_j - \mu$
and $\mu_{\ell} - \mu$ \emph{projected} along the direction of $\mu_j - \mu_{\ell}$ to decide wether or not to remove $j$ from the list. This kind of comparisons, essentially reduce the problem into  one-dimensional tasks, where we can again just use the $1$-d estimator to learn the projection of $\mu$ along the direction of $\mu_{\ell} - \mu_j$.
We remark that while the goal of a tournament procedure 
is to prune the candidate list down to a shorter list \emph{containing} 
a good estimate in the list-decodable setting, our tournament procedure 
needs to produce a \emph{single} estimate that has nearly optimal 
error. To achieve this, we require a novel analysis leveraging the 
structure of our setting.
The formal statement and proof can be found in \Cref{lem:prune}.

\subsection{Rejection Sampling}
\label{sec:rejection-sampling}
Recall that our procedure~\recursivefunction~relies critically on finding a low-variance subspace.
Unfortunately, a low-variance subspace 
may not exist in general---as even a single noisy sample 
can cause the variance to be arbitrarily large.
Our key idea is to use the warm-start estimate $\tilde \mu$ obtained from the tournament procedure 
as a rejection sampling center to filter out unusually noisy samples: 
If $x_i \sim \cN(\mu, \vec \Sigma_i)$ denote the original samples,
for each $i$  we sample a decision variable $b_i \sim \Bernoulli \left(\exp \left( - \|x_i - \meanini\|_2^2/d \right) \right)$ independently to decide whether or not to accept the sample.\looseness=-1

\paragraph{Distribution of Accepted Samples}
To analyze the effect of rejection sampling, denote by $\cA_i$ the distribution of the $i$-th sample conditioned on its acceptance, i.e., the distribution of $x_i \sim \normal(\mu, \vec \Sigma_i)$ conditioned on $b_i=1$.
A convenient fact (cf. \Cref{lem:product}) is that
$\cA_i$ is simply another Gaussian $\cN(\tilde \mu_i, \tilde {\vec \Sigma}_i)$, 
for some $\tilde \mu_i, \tilde {\vec \Sigma}_i$ defined as
functions of $\mu,\meanini,\vec \Sigma_i$ that are provided in \eqref{eq:gaussian_formula}. 
Denote by $\Sacc = \{i \in [n]: b_i=1 \}$ the set of all accepted indices (which is a random set). 
By independence between samples and \Cref{lem:product}, 
if we condition on a set $\Sacc=S$,
then the conditional joint distribution of $\{x_i\}_{i \in S}$
is equal to the product distribution of the Gaussians $\cN(\tilde \mu_i, \tilde {\vec \Sigma}_i)$ for $i \in S$ (with $\tilde \mu_i, \tilde {\vec \Sigma}_i$ defined in \eqref{eq:gaussian_formula}).
We can thus adopt the following more convenient view of the random process that generates the \emph{accepted} samples. 
\begin{definition}[Generation of accepted samples---alternative view]\label{def:data-gen-accepted}\phantom{x}
   \begin{enumerate}[leftmargin=*,itemsep=-2pt]
    \item $\Sacc$ is generated by independently including $i \in [n]$ 
    with probability $\E_{ x_i \sim \normal(\mu, \vec \Sigma_i)}[ e^{-  \| x_i - \tilde \mu\|_2^2 / d} ]$.
    \item For each $i \in \Sacc$,  $x_i$ is drawn independently from $\mathcal A_i := \cN(\tilde \mu_i, \tilde {\vec \Sigma}_i)$, where
    $\tilde \mu_i, \tilde {\vec \Sigma}_i$ are defined in \eqref{eq:gaussian_formula}.\looseness=-1
\end{enumerate} 
\end{definition}
\noindent Thus, if we condition on a particular $S_{\accept}$ 
of accepted indices, we could view the accepted vectors 
simply as independent samples from the distributions
$\{\mathcal A_i = \normal( \tilde \mu_i, \tilde {\vec \Sigma}_i )\}$. 
This property makes it manageable to analyze the statistical properties 
of subroutines that work on the accepted samples.

\paragraph{Rejection of Noisy Samples}
{Let $x_1,\ldots,x_k \in \R^d$ be the accepted samples (by renaming $S_{\accept}$ to $[k]$).}
On the one hand, by definition of our acceptance rule, we will rarely see samples that are at a distance more than $\sqrt{d}$ from the rejection center $\meanini$ {(where $d$ is the dimension of the subspace we are currently working in)}.
This ensures that the averaged covariance matrix of the accepted samples, 
$\frac{1}{k} \sum_{i=1}^k  \tilde {\vec \Sigma}_i$,
must have its trace appropriately bounded.
In particular, we show in \Cref{lem:expected-norm-tail} that with high probability over the randomness of the accepted indices $S_{\accept}$, the distributions $\mathcal A_i$ for $i \in S_{\accept}$ satisfy that
$
\E_{ z \sim \mathcal A_i }\lp[ \| z_i - \meanini \|_2^2 \rp] \lesssim d \log (nd).
$
Via a linear algebraic argument, 
this immediately implies that $\tr \lp( \frac{1}{k} \sum_{i=1}^k  \tilde {\vec \Sigma}_i \rp) \lesssim d \log (nd)$.
By a simple averaging argument, we therefore know that there must exist a subspace of dimension at least $d/2$ such that the eigenvalues of $\frac{1}{k} \sum_{i=1}^k  \tilde {\vec \Sigma}_i$ within the subspace is at most $ O \lp( \log(nd) \rp)$.
This therefore guarantees the existence of a low-variance subspace of the average population covariance matrix. Finally, note that $\tilde {\vec \Sigma}_i$ are unknown population covariances; thus, even though 
we showed that a low-variance subspace exists, 
we yet have to provide a way to compute 
such a subspace from samples. This will be done in \Cref{sec:search}.

\paragraph{Survival of Samples with Bounded Covariance}
On the other hand, by standard Gaussian norm concentration, a sample $x_i \in \R^d$ with covariance bounded by $\vec I$ is rarely at a distance more than $\sqrt{d}$ from the mean $\mu$.
Provided that $\| \meanini - \mu \| \lesssim \sqrt{d}$, we thus have that such samples will pass the rejection sampling procedure with at least constant probability. 
A slight issue is that as we go deeper into the recursion tree of the algorithm, each call of the recursive routine projects everything onto a subspace with half the dimension $d$ of the parent call. Thus, the requirement $\| \meanini - \mu \| \lesssim \sqrt{d}$, although true in the beginning when $d = D$, may no longer hold in later steps of the recursion.
Fortunately, at each call of the routine, we can take more fresh samples, project them onto the current subspace $\cV$, and invoke the tournament procedure from \Cref{lem:prune_informal} to produce some new estimate $\tilde \mu \in \R^d$ that is guaranteed to be within distance $\sqrt{d}$ from the projected mean (cf. \Cref{line:tournament} in \Cref{alg:recursive}).

\paragraph{Bias of Rejection Sampling}
{A more significant issue concerns the bias in the mean of the accepted samples. Let us examine the averaged mean over the accepted samples.}
A straightforward computation (cf. \Cref{lem:error_on_low_var}) shows that this expectation is:
\begin{align}
\label{eq:intro-bias-form}
\frac{1}{k}
\sum_{i=1}^k \tilde \mu_i
- \mu
=  \frac{2}{d} \Sigmavg \lp( 
\meanini - \mu \rp) \, ,    
\end{align}
where we define 
$ \Sigmavg := \frac{1}{k} \sum_{i=1}^k \vec {\tilde \Sigma}_i $ for convenience.
{The question then is what is the best upper bound that we can use for the operator norm of $\Sigmavg$.
A basic fact is that, conditioned on the accepted indices, each accepted sample $x_i$ follows the Gaussian $\cA_i := \cN(\tilde \mu_i, \tilde {\vec \Sigma}_i)$ with $\tilde \mu_i = \vec {\tilde \Sigma}_i \lp( 2\meanini / d + \vec \Sigma_i^{-1} \mu \rp) $
and $\vec {\tilde \Sigma}_i = \lp( \vec \Sigma_i^{-1} + (2/d)\vec I  \rp)^{-1}$ (cf. \Cref{eq:gaussian_formula}). This immediately implies that $\|\vec {\tilde \Sigma}_i\|_2 \leq d/2$.
However, if the operator norm of the averaged covariance $\Sigmavg$ 
could be as large as $d/2$, then the bias of the accepted samples would not be 
any better than the error of $\tilde \mu$ that we started with!}
Fortunately, as discussed earlier, the fact that $\Sigmavg$ has bounded trace implies that there must exist a subspace $\cV$ of dimension at least $d/2$ such that the eigenvalues of $\Sigmavg$ constrained to $\cV$ are at most $\log(nd)$.
Hence, the empirical mean over the accepted samples would constitute a good estimator within this low-variance subspace $\cV$---if we could successfully \emph{identify} it.
After that, we can recurse on the orthogonal complement of $\cV$, and that would complete the analysis of~\recursivefunction.
The remaining task is therefore to design a procedure that can identify this subspace.

\subsection{Searching for a Low-Variance Subspace}\label{sec:search}
We devise a procedure~\SearchSubSpace~(cf. \Cref{alg:find-subspace}) for identifying a low-variance subspace with respect to the accepted samples.
\begin{algorithm}    
\caption{Function to find a subspace in which the accepted samples have low variance}
\label{alg:find-subspace}
\begin{algorithmic}[1]
\Function{\SearchSubSpace}{$\tilde {\mu}, x_1, \cdots , x_k$}
\parState{\textbf{Input}: 
Rough mean estimate $\tilde \mu \in \R^d$, and 
samples $x_1, \cdots, x_k \in \R^d$. 
}
    \parState{\textbf{Output}: Matrices $\RowP_{\mathrm{low}},\RowP_{\mathrm{high}}  \in \R^{d/2 \times d}$ whose rows form an orthonormal basis of $\R^d$.\looseness=-1
    }
    \vspace{-5pt}
        \State Let $\Memp = \tfrac{1}{k} \sum_{i \in [k]} (x_i - \tilde \mu) (x_i - \tilde \mu)^\top$.
        \parState{
        Let $v_1, v_2, \cdots, v_d$ be the $d$ eigenvectors of $\Memp$ in descending order of their eigenvalues.}
        \parState{Let $\RowP_{\mathrm{high}} = [v_1,\ldots,v_{\floor{d/2}}]^\top$ and $\RowP_{\mathrm{low}} = [v_{ 
        \floor{d/2}+1},\ldots,v_d]^\top$.}
        \State \textbf{return} $\RowP_{\mathrm{low}},\RowP_{\mathrm{high}}, \mul$ 
\EndFunction
\end{algorithmic}
\end{algorithm}
\begin{lemma}[Low-Variance Subspace Identification (Informal; see \Cref{lem:low-var-identification})]
\label{lem:low-var-identification-informal}
Let $x_i \sim \normal( \tilde \mu_i, \tilde {\vec \Sigma}_i ) \in \R^d $ for $i \in [k]$ be the set of accepted samples.
Assume that $k \gg   d \polylog(nd)$.
The procedure~\SearchSubSpace~takes the samples as input, 
and produces a subspace $\cV$
such that the following holds with high probability:
$
v^\top \frac{1}{k} \sum_{i \in [k]} 
\tilde {\vec \Sigma}_i v
\lesssim  \log(nd)
$
for all unit vectors $v \in \cV$. 
\end{lemma}
\noindent To search for the low-variance subspace, we take the following natural approach. 
Given accepted samples $x_1, \cdots, x_k \in \R^d$ that pass through the rejection sampling centered at $\meanini$, we compute the second moment matrix $ 
\Memp := \frac{1}{k} \sum_{i=1}^k \lp( x_i  - \meanini \rp) \lp( x_i  - \meanini \rp)^\top $, and take the subspace spanned by the bottom $d/2$ eigenvectors.
The challenging part is to show that the averaged covariance matrix $\Sigmavg$ will indeed have small eigenvalues within this subspace with high probability.
The proof strategy comprises of {showing that the following two claims hold with high probability}:
\begin{enumerate}[leftmargin=*]
    \item There is a subspace $\cV$ in which the empirical second moment $\Memp$ of accepted samples has small operator norm.\looseness=-1\label{it:low-var-exists}
    \item The averaged covariance matrix $\Sigmavg$ {can be bounded from above (in Lowner order) by the} sample second moment matrix $\Memp$ (up to a constant factor). \label{it:cov-conc}
\end{enumerate}
It is not hard to see that combining the two {statements above} 
yields that, with high probability, $\Sigmavg$ will have bounded eigenvalues within the subspace $\cV$.
The formal statement of \Cref{it:low-var-exists}
is given in \Cref{cor:subspace-existence}. Its proof follows mostly from the properties of rejection sampling and the details are deferred to \Cref{sec:search-subspace}.
The formal statement of \Cref{it:cov-conc} is given in \Cref{lem:cov-concentration}.
The proof idea is to use the ``truncation'' technique. In particular, define $\vec X_i = \lp( x_i - \meanini \rp) \lp( x_i - \meanini \rp)^{\top} $. We will decompose $\vec X_i$ into
\begin{align*}
    \vec Y_i = \mathbbm 1 \{ \|  x_i - \meanini \|_2 \leq \tilde \Theta( \sqrt{d} ) \} \vec X_i \quad \text{and} \quad \vec Z_i = \mathbbm 1 \{ \|  x_i - \meanini \|_2 > \tilde \Theta( \sqrt{d} ) \} \;.
\end{align*}
By the property of the rejection sampling procedure, 
the $\vec Z_i$ are almost always $\vec 0$. 
On the other hand, the $\vec Y_i$ now have bounded operator norm almost surely. We could then apply the Matrix Bernstein Inequality to argue about the spectral concentration of $\sum_{i \in [k]} \vec Y_i$.
The details are provided in \Cref{sec:search-subspace}.

\subsection{Improvement within the Low-Variance Subspace}
Building on top of~\SearchSubSpace, we give a procedure $\textsc{PartialEstimate}$
{(cf. \Cref{alg:recursive})}
that simultaneously identifies a low variance subspace $\cV \subset \R^d$ and an estimate $\mul$ such that they satisfy \Cref{eq:goal} with high probability.
See \Cref{lem:low_var_error_impr} for the formal statement.
Notably, the routine uses the same batch of accepted samples to search for the low-variance subspace, and to compute the empirical mean estimator $\mul$.
To avoid adaptive conditioning (i.e., conditioning 
on concentration of the sample mean within the subspace $\cV$ 
defined by the samples), we instead show that 
with high probability the deviation of the sample mean from the population mean along \emph{every} direction is bounded from above by some quantity proportional to the population variance along that direction.
{That is, conditioned on that $\{1,\ldots,k\}$ are the indices of the accepted samples, then with probability at least $1-\tau$, the following holds for all $v \in \R^d$
\begin{align}\label{eq:mean-conc_everywhere}
v^\top \lp( \frac{1}{k} \sum_{i = 1}^k x_i - \frac{1}{k} \sum_{i =1}^k \tilde{\mu}_i \rp)
\lesssim \sqrt{\frac{1 }{d}}  \sqrt{ \frac{1}{k} \sum_{i =1}^k  v^\top \vec {\tilde\Sigma}_i v }.
\end{align}
This is because, conditioned on the acceptance set being $\{1,\ldots,k\}$, 
the accepted samples are just independent Gaussians conditioned on a 
specific set $S$ of accepted indices; 
thus, standard Gaussian concentration can be applied.} 
The detailed proof is given %
in \Cref{lem:overall-mean-concentration}. 

\begin{algorithm}
\caption{Function to estimate the mean within the Low-Variance Subspace}    
\label{alg:partial-estimate}
\begin{algorithmic}[1]
\Function{PartialEstimate}{$\tilde \mu, x_1, \cdots, x_n$}
\parState{\textbf{Input}:
Rough mean estimate $\tilde \mu \in \R^d$, and samples $x_1, \cdots, x_n \in \R^d$.}
\parState{\textbf{Output}: Matrix $\RowP_{\mathrm{high}}  \in \R^{d/2 \times d}$ and 
    an estimate $\mul \in \R^{d}$.
}\vspace{8pt}
         
        \State For each $i \in [n]$ draw $b_i \sim \mathrm{Ber}(e^{- \|x_i-\tilde \mu\|_2^2/d})$. \label{line:rejection_sampling} \Comment{Rejection sampling}
     
        \State $\RowP_{\mathrm{low}},\RowP_{\mathrm{high}} \gets \Call{\SearchSubSpace}{\tilde \mu, \{ x_i : b_i = 1 \}}$.\label{line:search_subspace}
        \Comment{Split space (cf. \Cref{sec:search})}
        \State $\mul \gets \tfrac{1}{\sum_{i \in [n]} b_i} \sum_{i \in [n]} b_i \RowP_{\mathrm{low}}^\top \RowP_{\mathrm{low}} x_i$  \label{line:mu1} \Comment{Empirical mean of surviving samples after projection}
        \State \textbf{return} $\RowP_{\mathrm{high}}, \mul$ 
\EndFunction
\end{algorithmic}
\end{algorithm}
\paragraph{Putting everything together to show \eqref{eq:goal}}
We are now ready to put everything together to prove \eqref{eq:goal}.
Recall that the procedure~\recursivefunction~works in a $d$-dimensional subspace of $\R^D$ defined by some orthonormal matrix $\vec P \in \R^{d \times D}$ given as its input.
After projecting the sample points into the subspace with $\vec P$, it runs the tournament procedure from \Cref{lem:prune_informal} to obtain a rough estimate $\tilde \mu \in \R^d$ within the subspace. 
It then performs rejection sampling as outlined in \Cref{sec:rejection-sampling} centered at $\tilde \mu$.
Let $x_1, \cdots, x_k$ be the set of accepted samples.
By \Cref{eq:intro-bias-form}, the bias of the samples will be of the form
$ \lp( 2 / d \rp) \Sigmavg \lp( \meanini - \mean \rp) $.
With high probability, the routine~\SearchSubSpace~yields a subspace $\cV$ such that $v^T \Sigmavg v = O(1) $ (up to polylogarithmic factors) for all $v \in \cV$.
Via a linear algebraic argument (cf. \Cref{lem:mean-projection-bound}), one can show that 
$  \left\| \vec \Pi_{\cV} \left(\tfrac{1}{k} \sum_{i=1}^k \E[ x_i ] - \mean\right)\right \|_2  \ll \| \meanini - \mean \|_2$.
On the other hand, as shown in \Cref{eq:mean-conc_everywhere}, the deviation of the empirical mean 
$\frac{1}{k} \sum_{i=1}^k  x_i$ from its expected value will be at most $\sqrt{d / k }$ (up to polylogarithmic factors).
Combining the above two observations with the triangle inequality then concludes the proof of \Cref{eq:goal}.

The detailed pseudocode of~\recursivefunction~is given in \Cref{alg:recursive}.
For the formal statement and proof, we refer the reader to \Cref{lem:low_var_error_impr}.

\begin{algorithm}[H]
\caption{Recursive Function for Mean Estimation}
\label{alg:recursive}
\begin{algorithmic}[1]
\Function{RecursiveEstimate}{$\RowP, n, \hat{\mu},\tau $}
\parState{\textbf{Input}: 
    Row orthonormal matrix $\RowP \in \R^{d \times D}$,
    sample budget $n \in \N$,
    failure probability $\tau \in (0, 1)$,
    mean estimate $\hat{\mu} \in \R^d$ from last iteration,
    and noise-to-signal ratio $\alpha \in (0,1)$.
    }
    \State{\textbf{Output}: 
    Better mean estimate $\hat{\mu}' \in \R^d$.}
    \vspace{5pt}
    \State Let $C, \delta$ be sufficiently large absolute constants.
    \parState{Draw an independent 
    batch of $n$ samples 
    $y_1, \cdots, y_n \in \R^D$ from the model of \Cref{def:model} (see \Cref{fnote:samples}). \label{line:dataset}}
    \State Form a projected set of samples: $S \gets \{ x_i = \RowP y_i  \; \forall i \in [n] \} \subset \R^d$. 
    \label{line:proj}
    \State $\tilde \mu \gets \Call{TournamentImprove}{\hat \mu,n,\tau}$ of \Cref{lem:guess-center}. \label{line:tournament} \Comment{Ensures $\|\tilde{\mu} {-} \vec P \mu \|_2 {\lesssim }\sqrt{d} {+} \f(\alpha,n)$.}
    \If{$d \leq C \varb (\L)^2 $} \label{line:base}
        \State Let $\hat{\mu}' \in \R^d$ be the output of the estimator from \Cref{cor:naive_multivariate} computed on the dataset $S$.\label{line:naive_est}
        \State \textbf{return} $\hat{\mu}'$ \Comment{End the recursion}
    \ElsIf{ $\sqrt{d} \leq \f(\alpha, n)$ }\Comment{Recall that $\f(\cdot)$ is defined in \Cref{eq:function_f}}
    \State \textbf{return} $\hat{\mu}' = \tilde \mu$. \label{line:base2}
    \Comment{We already have 
    $\| \meanini - \vec P \mean \|_2 \lesssim \; \f(\alpha, n)$ in this case
    }
    \Else  \label{line:else} 
        \parState{$\RowP_{\mathrm{high}},\mul  \gets $ \Call{PartialEstimate}{$\tilde \mu, x_1, \cdots, x_n$}\label{line:partialestimate} \Comment{(cf. \Cref{alg:partial-estimate}) Returns an orthonormal matrix
        $\RowP_{\mathrm{high}} \in \R^{d/2 \times d}$
        corresponding to the high-variance subspace and some mean estimate $\mul \in \R^d$ in the low-variance subspace.}   }
        \parState{$\muh \gets $\Call{RecursiveEstimate}{$\RowP_{\mathrm{high}}  \RowP , n, 
        \RowP_{\mathrm{high}} \hat \mu, \tau$} \label{line:recursive-call}\Comment{Recurse on high-variance subspace} } \label{line:mu2}
        \State $\hat{\mu}' \gets  \mul +  \RowP_{\mathrm{high}}^{\top} \muh$ \label{line:combine} \Comment{Combination of the two estimates in $\R^d$}
        \State \textbf{return} $\hat{\mu}'$
    \EndIf
\EndFunction

\end{algorithmic}
\end{algorithm}

\paragraph{Organization}
{The structure of the body of this paper is as follows:
\Cref{sec:preliminaries} records our notation conventions, basic facts, and the result of \cite{compton2024near} that will be used in later sections.
\Cref{sec:tournament} is dedicated to the warm-start estimate that is based on the tournament procedure.
\Cref{sec:small_error_in_low_var}, which is the bulk of the technical work, gives
the analysis for the low-variance subspace that the algorithm identifies.
Finally, \Cref{sec:proof_of_thm} puts everything together to establish 
\Cref{thm:main}.}

\section{Preliminaries}\label{sec:preliminaries}

\subsection{Notation and Useful Facts}

\paragraph{Notation}
We use $\mathbb{Z}_+$ for the set of positive integers and $[n]$ to denote $\{1,\ldots,n\}$. For a vector $x$ we denote by $\|x\|_2$ its  Euclidean norm. Let $\bI_d$  denote the $d\times d$ identity matrix (omitting the subscript when it is clear from the context). 
We use $\cS^{d-1}$ to denote the set of points $v \in \R^d$ with $\|v\|_2=1$.
We use  $\top$ for the transpose of matrices and vectors.
For a subspace $\mathcal{V}$ of $\R^d$ of dimension $m$, we denote by $\vec \Pi_{\cV} \in \R^{d \times d}$ the orthogonal projection matrix of $\cV$. 
 That is, if the subspace $\cV$ is spanned by the rows of the matrix $\bA$, then $ \vec \Pi_{\cV}:=\bA^\top(\bA \bA^\top)^{-1} \bA$ (if the matrix $\bA$ is row orthonormal then we have the simpler expression $\vec \Pi_{\cV}:=\bA^\top \bA$). 
We say that a symmetric $d\times d$ matrix $\bA$ is PSD (positive semidefinite) and write $\bA \succeq 0$ if for all $x\in \mathbb{R}^d$ it holds $x^\top \bA x\ge 0$. We use $\|\bA\|_{2}$ for the operator (or spectral) norm of the matrix $\bA$. We use $\tr(\bA)$ to denote the trace of the matrix $\bA$.

We write $x\sim {\cD}$ for a random variable $x$ following the distribution $\cD$ and use $\E[x]$ for its expectation. We use $\cN(\mu,\vec \Sigma)$ to denote the Gaussian distribution with mean $\mu$ and covariance matrix $\vec \Sigma$. We write $\pr[\cE]$ for the probability of an event $\cE$. 
We use $\Bernoulli(p)$ to denote 
the Bernoulli distribution with mean $p$.
We use $\cU(S)$ to denote the uniform distribution over a set $S$.

We use $a\lesssim b$ to denote that there exists an absolute universal constant $C>0$ (independent of the variables or parameters on which $a$ and $b$ depend) such that $a\le Cb$.
Sometimes, we will also use the $O(\cdot),\tilde O(\cdot),
\Omega(\cdot),\tilde \Omega(\cdot),
\Theta(\cdot),\tilde \Theta(\cdot)$ notation with the standard meaning (in particular the $\tilde O,\tilde \Omega,\tilde \Theta$ notation hides polylogarithmic factors).
We use the notation $a \gg b$ to denote that $a > C b$ for a sufficiently large absolute constant $C$. We will use $\polylog(\cdot)$ to denote a quantity that is polylogarithmic in the arguments, with the degree being a sufficiently large absolute constant. {We use $\log_2$ for the logarithm with base $2$ and just $\log$ for the natural logarithm.}

We begin this section with an elementary analysis fact that will be useful in calculations of later sections.
\begin{fact}[See, e.g., Lemma A.3 in \cite{shalev2014understanding}]\label{fact:implicitinequality}
        Let $r \geq 1$ and $b>0$. Then $n \geq 4 r \log(2r) + 2 b $ implies that $ n \geq r \log n + b$.
    \end{fact}

We also state two probability facts that will be useful.
\begin{fact}[Matrix Bernstein Inequality, see, e.g., Theorem 5.4.1 in \cite{Vershynin18}]\label{fact:matrixBernstein}
Consider a finite sequence $\{ \vec X_i \}_{i=1}^k$  of independent, random,  $d \times d$ Hermitian matrices.
Assume that $ \E [ \vec X_i ] = \vec 0$
and $\| \vec X_i \|_2 \leq L$ for all $i \in [k]$.
Define $\nu = \lp \| \sum_{i=1}^k \E[  \vec X_i^2 ] \rp \|_{2}$.
Then it holds that
\begin{align*}
    \Pr \lp[ \lp \|  \frac{1}{k} \sum_{i=1}^k \vec X_i \rp \|_2 
    \geq t  \rp]
    \leq d \exp \lp( 
    \frac{ -k^2 t^2 / 2 }{ \nu + L kt / 3 }
    \rp).
\end{align*}
\end{fact}

\begin{fact}[Gaussian Rejection Sampling, see, e.g., Section 8.1.8 in \cite{petersen2008matrix}]
\label{lem:product}
    Denote by $p_{\cN(\mu_1, \vec \Sigma_1)}(x)$ and $p_{\cN(\mu_2,\vec \Sigma_2)}(x)$ the pdf functions of the Gaussians $\cN(\mu_1, \vec \Sigma_1)$ and $\cN(\mu_2, \vec \Sigma_2)$ respectively.
    Then,
    \begin{align*}
        p_{\cN(\mu_1, \vec \Sigma_1)}(x) \cdot p_{\cN(\mu_2, \vec \Sigma_2)}(x) 
        \propto \cN(\tilde \mu,\vec {\tilde \Sigma}) \;,
    \end{align*}
    where $\tilde \mu$ and $\vec {\tilde \Sigma}$ are defined as
    \begin{align*}
    \vec {\tilde \Sigma} = (\vec \Sigma_1^{-1} + \vec \Sigma_2^{-1})^{-1} \quad \text{and} \quad
        \tilde \mu_i =   \vec {\tilde \Sigma}
         \left(   \vec \Sigma_1^{-1} \mu_1 + \vec \Sigma_2^{-1} \mu_2 \right) \;.
    \end{align*}
\end{fact}

\noindent In particular, we will use that fact later on with $\mu_1 = \meanini, \vec \Sigma_1 = (d/2) \;\vec I, \mu_2 = \mu,\vec \Sigma_2 = \vec \Sigma_i$. We then have that $p_{\cN(\meanini, \vec I)}(x) \cdot p_{\cN(\mu,\vec \Sigma_i)}(x) \propto \cN(\tilde \mu_i,\vec {\tilde \Sigma}_i)$ with 
\begin{align}\label{eq:gaussian_formula}
    \vec {\tilde \Sigma}_i = (\vec \Sigma_i^{-1} + \tfrac{2}{d}\vec I )^{-1} \quad \text{and} \quad
        \tilde \mu_i =   \vec {\tilde \Sigma}_i
         \left(   \tfrac{2}{d} \meanini + \vec \Sigma_i^{-1} \mu \right) \;.
\end{align}

\paragraph{Data Generation Model}
{In the analysis of the algorithm, it will be convenient to assume that the covariances of the data have their eigenvalues bounded from below so that we do not need to deal with degenerated data instances.
We remark that we can assume the eigenvalue lower bound because otherwise we can enforce it by a simple transformation of the samples.}
\begin{remark}[Covariance Eigenvalue Lower Bound]
\label{remark:cov-lb}
The original data generation model does not enforce that the covariance matrices of the data have full rank. However, we can reduce to the case that $\vec \Sigma_i \succeq \frac{1}{2} \vec I$ 
by introducing a preprocessing step that transforms each data point $x$ via $x'= \lp(x + y \rp)/\sqrt{2}$, where $y \in \R^d$ is some standard Gaussian vector.
{After this transformation, each point has common mean $\mu' = \mu/\sqrt{2}$, and at least $\alpha N$ many of the points have  covariances $\vec \Sigma_i'$ satisfying $\tfrac{1}{2}\vec I \preceq  \vec \Sigma_i' \preceq \vec I$.}
\end{remark}

We conclude this section with a fact about the sampling process of the data generation model. 
For simplicity, we will state our algorithm as if it has access to the following data-generation model which is almost identical to the original.  The difference is that we now assume the algorithm can draw independent batches of samples, with all batches having the same signal-to-noise ratio and a common mean {as in the original model of \Cref{def:model}.}
\begin{restatable}[Data-generation model; independent batches]{definition}{MULTIBATCHMODEL} \label{def:model2}
    The algorithm can specify positive integers $t$ and $n_1,n_2,\ldots,n_t$ and receive datasets $S_1,S_2,\ldots,S_t$ such that the $S_i$'s are independent, each dataset $S_i$ has size at least $n_i$ and each $S_i$ follows the distribution defined in \Cref{def:model}.
\end{restatable}

In reality, the algorithm can only choose some integer $N$ and request a dataset of size $N$ satisfying the definition of \Cref{def:model}.
However, it is not hard to see that 
one can simulate access to, say $t$, independent batches of size (approximately) $n$ each, by requesting an original dataset of size $N = nt$, and then splitting it by assigning each sample to a random batch.
Conditioned on the partition, the samples within different batches remain independent from each other.
Moreover, 
if $n \gg \log(t/\tau)/\alpha$, then
with probability at least $1 - \tau$ it holds that each batch has size $\Theta(n)$ and the number of samples within each batch having their covariances bounded by $\vec I$ is at least $\alpha n$. We defer the proof of this fact to \Cref{sec:additional-prelims}.

\begin{restatable}{lemma}{SIMULATION}\label{lem:simulation}
    Let $t,n \in \Z_+$ and $\alpha,\tau \in (0,1)$. 
    Assume that $n \gg \log(t/\tau)/\alpha$.
    There is an algorithm which given sample access to the data generation model of \Cref{def:model} with common mean $\mu$ and signal-to-noise rate $\alpha$, takes  {$t,n$} as input, draws $N = tn$ samples, and outputs $t$ independent sets of samples such that with probability at least $1-\tau$ the output is distributed according to the model of \Cref{def:model2} with common mean $\mu$, signal-to-noise-{rate} at least $0.9\alpha$, and the size of each batch is at least $0.9n$.
\end{restatable}

Note that having a $0.9\alpha$-fraction of good samples (instead of the original $\alpha n$) is not a significant disadvantage, as it only affects the constants in our bounds wherever it is applied.

\subsection{{Low-Dimensional Mean Estimation in the Subset-of-Signals Model}}

We start with the estimator from \cite{compton2024near}, which works for one-dimensional datasets.

\begin{theorem}[\cite{compton2024near}]\label{thm:one-dim}
{Let $\delta \in [0,\infty)$ be an absolute constant.}
There exists an algorithm which, having as input $n$ samples in $\R$ from the model of \Cref{def:model2} with common mean $\mu \in \R$ and signal-to-noise rate $\alpha \in (0,1)$, outputs an estimate $\hat \mu \in \R$ such that {with probability at least $1-1/n^{\delta}$}, it holds $|\hat \mu - \mu| \leq \f(\alpha,n)$, where $\f(\cdot)$ is defined as follows:
    \begin{align}\label{eq:function_f}
    \f(\alpha,n) = C_\delta \cdot (\log(n/\alpha))^{O(1)} \cdot \begin{cases}
        \displaystyle \frac{1}{\alpha^2 n^{3/2}},  &\displaystyle C_\delta\frac{\log n}{n}  \leq \alpha \leq \frac{1}{n^{3/4}} \\[11pt]
        \displaystyle \frac{ 1 }{\alpha^{2/3} n^{1/2}}, &\displaystyle \frac{1}{n^{3/4}} < \alpha < 1 \\[11pt]
        \infty , &\text{otherwise} \, ,
    \end{cases}
\end{align}
{where $C_{\delta}$ is some constant that depends only on $\delta$.}
\end{theorem}

By running the above estimator along every axis in $\R^d$ (where each run shares the same dataset), this gives the following naive extension of the estimator to the space $\R^d$.
\begin{corollary} \label{cor:naive_multivariate}
Let $\delta \in [0,\infty)$ be an absolute constant.
There exists an algorithm which, having as input $n$ samples from the model of \Cref{def:model2} in $\R^d$ with common mean $\mu \in \R^d$ and signal-to-noise ratio $\alpha \in (0, 1)$, outputs an estimate $\hat{\mu} \in \R^d$ such that with probability at least $1-d/n^\delta$ it holds $\|\hat{\mu} - \mu\|_2 \leq \f(\alpha,n) \sqrt{d}$, where $\f(\cdot)$ is the function defined in \Cref{eq:function_f}.    
\end{corollary}

\section{Warm Start: Rough Estimate via Tournament}\label{sec:tournament}
{In this section, we first present the tournament procedure that selects (approximately) the best mean estimate among a given candidate list.
Building on top of it, we construct the warm-start routine 
\textsc{TournamentImprove} that can produce an initial rough estimate $\mu$ with error $O ( \sqrt{d} + f(\alpha, N) )$.
}

\begin{lemma}[Tournament]
\label{lem:prune}
Let $\delta \in [0,\infty)$ be an absolute constant.
Let $n \in \N$,  $L = \{ \mu_1, \cdots, \mu_k \} \subset \R^d$ be a set of candidate estimates of $\mu \in \R^d$. 
There exists an algorithm \textsc{TournamentImprove} that takes
the list $L$ as input, 
draws a dataset of size $n$
according to the data generation model of \Cref{def:model2}
with common mean $\mu \in \R^d$ and signal-to-noise ratio $\alpha$, 
runs in time $\poly(n,k, d)$, and outputs some estimate $\mu_{j} \in L$ such that
$ \| \mu_{j} - \mu \|_2 \leq 
2 \min_{i \in [k]} \| \mu_i - \mu \|_2 + 4 \f(\alpha, n)$ with probability at least $1- k^2/n^\delta$, where $\f(\alpha, n)$ is defined as in \Cref{thm:one-dim}.
\end{lemma}
\begin{proof}
We will run a tournament among the candidate estimates.
To compare two candidate means $\mean_j, \mean_{\ell}$, we will project the samples along the direction of $v_{j, \ell}:= \lp( \mean_{\ell} - \mean_j \rp) / \| \mean_{\ell} - \mean_j \|_2 $, and then run the $1$-d mean estimation algorithm.
By \Cref{thm:one-dim}, the algorithm produces some number $\tilde \mean_{j, \ell} \in \R$
such that, 
\begin{align*}
|\tilde \mean_{j, \ell} - v_{j, \ell}^\top \mean | \leq \f(\alpha, n)
\end{align*}
with probability at least $1-1/n^\delta$.
By the union bound, the above holds simultaneously for all $j, \ell \in [k]$ with probability at least $ 1 - k^2/n^{\delta}$.
We will condition on the above error bound in each direction
$v_{j, \ell}$ for $j \neq \ell \in [k]$.
This together with the triangle inequality immediately implies that
\begin{align}
\label{eq:1d-error-bound}  
| v_{j, \ell}^\top \mu_j  - v_{j,\ell}^\top \mu | - \f(\alpha, n) \leq
    | v_{j, \ell}^\top \mu_j - \tilde \mu_{j, \ell} |
    \leq | v_{j, \ell}^\top \mu_j  - v_{j,\ell}^\top \mu | + \f(\alpha, n).
\end{align}

In the tournament, we will disqualify $\mean_j$ if and only if
there exists some $\ell \in [k]$
such that $| v_{j, \ell}^\top \mean_j - \tilde \mean_{j, \ell}|
> | v_{j, \ell}^\top \mean_{\ell} - \tilde \mean_{j, \ell}| + 2 \f(\alpha, n)$.
Among the survival candidates, we pick one arbitrarily as the final output.

Suppose $\| \mean_j - \mean \|_2 = \min_{i \in [k]}  \| \mean_i - \mean \|_2$.
We will show that $\mean_j$ will not be disqualified.
Consider a $\mean_{\ell}$ for $\ell \neq j$.
For convenience, 
we denote by $v$  the unit vector pointing from $\mean_j$ to $\mean_{\ell}$, {and $\mean_{v^\perp} \in \R^d$ be the vector that: (i) lies on the plane defined by the three points $\mu_j,\mu_\ell,\mu$ and (ii) is perpendicular to the line connecting $\mu_j$ and $\mu_\ell$.
By the Pythagorean theorem we then have that}
$$
\lp(   v^\top \mean_j - v^\top \mean  \rp)^2
+ \| \mean - \mean_{v^\perp}  \|_2^2
= 
\| \mean_j - \mean \|_2^2
\leq 
\| \mean_{\ell} - \mean \|_2^2
=
\lp(   v^\top \mean_{\ell} - v^\top \mean  \rp) ^2
+ \| \mean - \mean_{v^\perp}  \|_2^2 \, ,
$$
which further implies that
\begin{align}
\label{eq:projection-optimality}
| v^\top \mean_j - v^\top \mean| 
\leq |  v^\top \mean_{\ell} - v^\top \mean |.    
\end{align}
Combining \Cref{eq:1d-error-bound,eq:projection-optimality} then gives that
\begin{align}
  |v^\top \mean_j - \tilde \mean_{j, \ell}| 
&\leq |v^\top \mu_j - v^\top\mu| + \f(\alpha,n) \tag{by \Cref{eq:1d-error-bound}}\\
&\leq |v^\top \mu_\ell - v^\top\mu| + \f(\alpha,n) \tag{by \Cref{eq:projection-optimality}}\\
&\leq  |v^\top \mean_{\ell} - \tilde \mean_{j, \ell}| + 2\f(\alpha, n) , \tag{by   \Cref{eq:1d-error-bound}}
\end{align}
implying that  $\mean_{j}$ will \emph{not} be disqualified.
This ensures that the survival candidate set will be non-empty.

Next, we show that any $\mean_{\ell}$ with $ \| \mean_{\ell} - \mean_j\|_2 > 2 \| \mean_j - \mean\|_2 + 4  \f(\alpha, n)$ will be disqualified. To show this, let $\mu_\ell$ be such a point.
First, since projection only shrinks distance between points, we have that
\begin{align}\label{eq:temp_ineq}
    | v^\top \mu - v^\top \mean_j | \leq \|  \mu - \mean_j \|_2.
\end{align}
Second, we have the following inequalities:
\begin{align*}
    | v^\top \mu_\ell - \tilde \mu_{j, \ell} | 
    &\geq | v^\top \mu_j - v^\top \mu_\ell | - | v^\top \mu_j - \tilde \mu_{j, \ell} | \tag{by the triangle inequality}\\
    &= \| \mu_j - \mu_\ell\|_2 - | v^\top \mu_j - \tilde \mu_{j, \ell} | \tag{since $v$ is a unit vector parallel to $\mu_j-\mu_\ell$}\\
    &> 2 \| \mu_j - \mu\|_2 + 4 \f(\alpha, n) - | v^\top \mu_j - \tilde \mu_{j, \ell} | \tag{by assumption}\\
    &\geq  2 \| \mu_j - \mu\|_2 + 3 \f(\alpha, n) - | v^\top \mu_j - v^\top \mu |  \tag{by \Cref{eq:1d-error-bound}}\\
    &\geq  \| \mu_j - \mu\|_2 + 3 \f(\alpha, n)  \tag{by \Cref{eq:temp_ineq}} \\
    &\geq |  v^\top \mu_j - v^\top\mu |+ 3 \f(\alpha, n) \tag{by \Cref{eq:temp_ineq}} \\
    &\geq | v^\top \mu_j - \tilde \mu_{j,\ell} |+ 2 \f(\alpha, n)   \;.  \tag{by \Cref{eq:1d-error-bound}}
\end{align*}
This ensures that $\mu_{\ell}$ will be disqualified.
Hence, the final surviving set will be non-empty and contains only $\mu_{\ell}$ such that 
$\| \mu_{\ell} - \mu_{h} \|_2 \leq 2 \min_{i \in [k]} \|  \mu_i - \mu \|_2 + 4 \f(\alpha, n)$.
This concludes the proof of \Cref{lem:prune}.
\end{proof}

Using the above tournament procedure, we can design a routine that takes an input vector $\hat \mu$, and 
with high probability produces some $\meanini$ such that
the distance from $\meanini$ to $\mu$ is (roughly) bounded from above by the minimum of $\|\hat \mu - \mu\|_2$ and $\sqrt{d}$.

\begin{lemma}[Tournament Improvement]
\label{lem:guess-center}
Let $\delta \in (2,\infty)$ be an absolute constant.
Let $n \in \N$, $\alpha \in (0, 1)$, $\hat \mu \in \R^d$.
There is an algorithm, \Call{TournamentImprove}{$\hat \mu,n$}, for which the following hold.
The algorithm draws two independent batches of size $n$ {according to the} data generation model of \Cref{def:model2} with common mean $\mean$ and signal-to-noise rate $\alpha$. 
The algorithm runs in $\poly(n,k,d)$ and
outputs some $\meanini \in \R^d$ such that, with probability at least $1-O(n^{2-\delta})$, it holds $\| \meanini - \mu\|_2 \lesssim 
    \min \lp(  \| \hat \mu - \mu\|_2, \sqrt{d} \rp) + \f(\alpha,n)$ where $\f(\alpha,n)$ is defined in \Cref{eq:function_f}.
\end{lemma}
\begin{proof}
Let $\tau:= 1/n^\delta$.
We can assume that 
$n \gg (\delta / \alpha) \log(\delta / \alpha)$. This is because otherwise $\f(\alpha, n)$ will be in the third regime where the function evaluates to $\infty$.
Then the algorithm could return any $\tilde \mu \in \R^d$.
By \Cref{fact:implicitinequality} and the assumption $n \gg (\delta / \alpha) \log(\delta / \alpha)$, we can assume that $n \gg \log(1/\tau) / \alpha$.
Let $x_i$ for $i \in [n]$ be the samples contained in the first batch of size $n$.
We begin by arguing that 
with probability at least $1 - \tau / 2$
there exists some sample $x_i$ such that 
$\| x_i - \mu \|_2 \leq 2 \sqrt{d}$.
Note that by \Cref{def:model2} and the assumption that $n \gg \log(1/\tau) / \alpha$, the batch contains at least $ 10 \log(1/\tau) $ many samples whose covariances are all bounded from above by $\vec I$.
{Denote by $S$ the set of the indices of these samples.}
Fix one such sample {$x_i$ for $i \in S$.} 
By Markov's inequality, if $\cE_i$ denotes the event $\|x_i - \mu\|_2 \geq 2 \sqrt{d}$ we have that $\Pr[\cE_i] \leq 1/2$. Since the events are independent, $\Pr[\cap_{i \in S}\cE_i] \leq 1/2^{ 10 \log(1/\tau) } \leq \tau/2$.
This means that 
\begin{align*}
    \Pr\left[ \exists i \in [n]: \| x_i - \mu \|_2 \leq 2 \sqrt{d} \right] \geq 1-\tau/2 \;.
\end{align*}

Conditioned on the above event, 
we will apply \Cref{lem:prune} with 
the input list $$L = \{ x_1, \cdots, x_n,   \}
\cup \{\hat \mu \}\;.
$$
The algorithm from \Cref{lem:prune} will use 
the remaining  batch of size $n$ samples and produce an output $\meanini$ such that
\begin{align*}
\| \meanini - \mean \|_2
&\lesssim \min\lp(  \sqrt{d}, \| \hat{\mu} - \mean \|_2 \rp) +   \f\lp(\alpha, n  \rp)
\end{align*}
with probability at least $1-O(n^{2 - \delta})$.
The overall success probability is at least $1-O(n^{2 - \delta})-\tau/2 = 1-O(n^{2 - \delta})$.
This concludes the proof of \Cref{lem:guess-center}.
\end{proof}

\section{{Obtaining Small} Error within Low-Variance Subspace}\label{sec:small_error_in_low_var}
This section includes the detailed analysis for the main algorithmic component \Call{RecursiveEstimate}{$\cdot$} of \Cref{alg:recursive}. 
In particular,
given some estimate $\meanini$, the routine simultaneously identifies a low-variance subspace and computes an estimate $\hat \mu'$ that has improved estimation error within the subspace. {That is, the main result that will be shown in this section is the following:}
\begin{restatable}[Improvement within Low-Variance Subspace]{lemma}{ERRORIMPROVEMENT}\label{lem:low_var_error_impr}
Let $x_1, \cdots, x_n \in \R^d$ be independent samples following the data generation model of \Cref{def:model2} with 
common mean $\mu \in \R^d$ and signal-to-noise rate
$\alpha \in (0, 1)$.
Let $\tilde \mu \in \R^d$ be some rough estimate satisfying $\| \tilde \mu - \mu \|_2 \lesssim \sqrt{d}$.
Assume that $n \gg \tfrac{d}{\alpha} \log(\tfrac{d}{\tau})\log(\tfrac{d}{\alpha \tau})$ for some $\tau \in (0,1)$
and $d$ is an even integer satisfying $d \gg \kappa^2 \varb$ for some $\kappa > 0$.
Then an execution of 
\Call{PartialEstimate}{$\tilde{\mu},x_1, \cdots, x_n$} defined in \Cref{alg:recursive}
produces a row orthonormal matrix
$\RowP_{\mathrm{low}} \in \R^{d/2\times d}$
and a vector $\mul \in \R^{d}$ (\Cref{line:search_subspace})
such that the following holds with probability at least $1 - O(\tau)$:
\begin{align*}
\lp \| 
\mul - \RowP_{\mathrm{low}}^\top \RowP_{\mathrm{low}} \mu
 \rp \|_2
\leq \frac{1}{\kappa}  \lp \|    \meanini - \mu   \rp\|_2
+ O\left( \sqrt{\frac{(d  +\log(1/\tau))\log(nd/\tau)}{\alpha n}} \right).
\end{align*}
\end{restatable}
{The section is structured as follows: In \Cref{sec:niceness_accepted}, we show the properties of the accepted samples that are sufficient for the rest of the analysis to hold. \Cref{sec:search-subspace,sec:bias} contain the rest of the analysis. Namely, \Cref{sec:search-subspace} shows that the routine~\SearchSubSpace~from \Cref{alg:recursive} will successfully identify a subspace $\cV$ such that the variance of the accepted samples along any direction within the subspace is low. \Cref{sec:bias} analyzes the bias and concentration of the accepted samples' mean in that low-variance subspace. \Cref{sec:mean_improvement} puts these two parts together to show \Cref{lem:low_var_error_impr}.}

Note that the routine will be recursively invoked within a chain of vector subspaces with different dimensions $d$.
In particular, in each recursive call of \Cref{alg:recursive}, the routine always draws samples $y_1, \cdots, y_n$ from the original $D$-dimensional space,  then immediately projects them onto some $d$-dimensional subspace, i.e., $x_i \gets \RowP y_i$ for some row orthonormal matrix $\RowP \in \R^{d \times D}$, and then performs all the operations on $x_i \in \R^d$.
For simplicity, {the next two subsections} present the analysis assuming that $D=d$ and $\RowP$ is the identity matrix (i.e., we analyze the first call in the recursion tree).

\subsection{Niceness of the Acceptance Set \texorpdfstring{$S_\accept$}{S-accept}}\label{sec:niceness_accepted}
In this subsection, we show that with high probability, the 
set of indices $S_{\accept}$ of samples accepted
by the rejection sampling procedure inside \Cref{alg:recursive} (which works by accepting the index $i$ with probability $ \exp( -  \| x_i - \meanini \|^2_2 /d)$ for some given rejection center $\meanini$) is a ``nice'' set. 
In particular, we will show that 
(a) the accepted vectors on average have small $\ell_2$ norm,  
(b) the distributions of the $\ell_2$ norm of the accepted vectors have  ``light'' tails,
(c) the second moment matrices (centered at the rejection center $\meanini$) of the accepted vectors have no trivial eigenvalues,
and
(d) sufficiently many vectors survive through the rejection sampling procedure.
The formal definition of a nice set is as follows.
\begin{definition}[Nice Accepted Set]
\label{eq:event_calibration}
Let $x_1, \cdots, x_{n}$ be samples generated by the model of \Cref{def:model} with signal-to-noise ratio $\alpha \in (0,1)$, and $\meanini \in \R^d$ be a fixed vector.
\begin{itemize}
    \item Let $b_i \in \{0,1\}$ for $i \in [n]$ be drawn from $\mathrm{Ber}(\exp( -  \| x_i - \meanini \|_2^2 / d ))$.
    \item Let $\mathcal A_i$ be the distribution of $x_i$ conditioned on $b_i = 1$. 
\end{itemize}
Then we define the family of nice sets 
to be as follows:
\begin{align} 
    \Snice = \bigg\{ 
    &S \subseteq [n] \; : \; \sum_{i \in S} 
    \Pr_{z \sim \mathcal A_i}
    \lp[ \| z - \meanini \|_2^2 \geq \pnormb \rp]
    \leq \tau \, , \nonumber \\
    &\E_{z \sim \mathcal A_i}
    \lp[ 
    \mathbbm 1 \{ z \geq \normb \}
    \| z - \meanini  \|_2^2
    \rp]
    \leq o(1/n) \;\; \forall i\in S \, , \nonumber \\ 
    &\E_{z \sim \mathcal A_i}
    \lp[ 
    \lp( z - \meanini  \rp)
    \lp( z - \meanini  \rp)^\top
    \rp]
    \succeq \vec I / 3  \;\;\forall i\in S\, , \nonumber \\     
    &|S| \geq \max\{C d \log(n/\tau)\log(d/\tau), \alpha n/C \}
    \bigg\} \, ,
\end{align}
where $C$ is some sufficiently large constant.
\end{definition}

In the next two lemmas, we start with the first two properties inside the definition of $\Snice$, which are all about the distributions of the $\ell_2$ norms of the surviving samples.
Notably, these two properties are independent of the initial distributions from which the samples are drawn.
\begin{lemma}[Norm Calibration via Rejection Sampling]\label{lem:norm_calibaration}
Let $x_1, \cdots, x_{n}$ be independent random variables in $\R^d$ following distributions $\mathcal D_1,\ldots,\mathcal D_n$, and let $\meanini \in \R^d$ be a fixed vector.
\begin{itemize}
    \item Let $b_i \in \{0,1\}$ for $i \in [n]$ be drawn from $\mathrm{Ber}(\exp( - \| x_i - \meanini \|_2^2 / d ))$.
    \item Let $\mathcal A_i$ be the distribution of $x_i$ conditioned on $b_i = 1$. 
\end{itemize}
Then the following holds with probability at least $1-\tau$.
\begin{align}
\label{eq:indicator-constraint}
    \sum_{i=1}^n 
    b_i
    \Pr_{z_i \sim \mathcal A_i}
    \lp[ \| z_i - \meanini\|_2^2 > 
    2 d\log(n/\tau) \rp]
    \leq \tau \;.
\end{align}
\end{lemma}
\begin{proof}
We first bound from above the expected value of the left hand side of \Cref{eq:indicator-constraint}.
In particular, we have that
\begin{align*}
\E[b_i] \Pr_{ z_i \sim \mathcal A_i }
\lp[  \| z_i - \meanini \|_2^2 > 2 d \log(n / \tau) \rp]
&= \Pr[ b_i = 1 ]
\Pr_{ x_i \sim \mathcal D_i} \lp[ \|x_i - \meanini \|_2^2 > 2 d \log(n / \tau) | b_i = 1  \rp] \\
&= 
\Pr_{  x_i \sim \mathcal D_i } \lp[ \|x_i - \meanini  \|_2^2 > 2 d \log(n / \tau) , b_i = 1  \rp] \\
&= 
\E_{  x_i \sim \mathcal D_i } \lp[ 
\mathbbm 1\{ \|x_i - \meanini  \|_2^2 > 2 d \log(n / \tau) \} 
\; \exp\lp( -  \|  x_i  - \meanini \|_2^2 / d \rp)
\rp] \\
&\leq \E_{  x_i \sim \mathcal D_i } \lp[  
\; \exp\lp( - 2 d \log(n / \tau) / d \rp)
\rp] \leq (\tau / n)^2.
\end{align*}
It then follows that the expected value of the left hand side of \Cref{eq:indicator-constraint} is at most $\tau^2 $. 
Thus, \Cref{eq:indicator-constraint} must hold with probability at least $1 - \tau$ by Markov's inequality.
\end{proof}

\begin{lemma}[Expected Norm Tail Bound]
\label{lem:expected-norm-tail}
Let $x_1, \cdots, x_{n}$ be independent random variables in $\R^d$ following distributions $\mathcal D_1,\ldots,\mathcal D_n$, and let $\meanini \in \R^d$ be a fixed vector.
\begin{itemize}
    \item Let $b_i \in \{0,1\}$ for $i \in [n]$ be drawn from $\mathrm{Ber}(\exp( -  \| x_i - \meanini \|_2^2 / d ))$.
    \item Let $\mathcal A_i$ be the distribution of $x_i$ conditioned on $b_i = 1$. 
\end{itemize}
Then the following holds with probability at least $1-\tau$.
\begin{align}
\label{eq:tail-expected-constraint}
\E_{ z_i \sim \mathcal A_i }\lp[ \mathbbm 1 \{ \| z_i - \meanini \|_2^2 \geq \normb \} \| z_i - \meanini \|_2^2 \rp] \leq o(1/n)
\;  \forall i \text{ such that } b_i = 1.
\end{align}
\end{lemma}
\begin{proof}
{Denote by $\mathcal I$ the set of indices such that $\Pr[b_i=1] \leq \tau / n$.
By the union bound, we have that 
    $\Pr\left[ \exists i\in \cI \; : \; b_i=1  \right]
    \leq \sum_{i \in \cI} \Pr[b_i=1] \leq \tau$.}
We have thus shown that the following event happens with probability at least $1-\tau$:
\begin{align}
\label{eq:no-rare}
b_i = 0 \; \forall i \text{ such that } \Pr[b_i=1] \leq \tau / n.    
\end{align}
In what follows, we condition on the above event.
Consider some $i$ such that $b_i = 1$, {which implies that
$\Pr[b_i=1] \leq \tau / n$ under the conditioning.}
Define $\mathcal G(z, t)$ as the event that
$ 2^t \anormb \leq \| z - \meanini  \|_2^2 \leq 2^{t+1} \anormb$, and $\rho_{\mathcal A_i}: \R \mapsto \R$ as the probability density function of the distribution of $\| z - \tilde \mu \|_2^2$, where $z \sim \mathcal A_i$.
We then have that
\begin{align}
&\E_{ z \sim \mathcal A_i }    
\lp[  \mathbbm 1 \{ \| z - \meanini \|_2^2 \geq \normb  \}  \| z - \meanini \|_2^2 \rp] \nonumber \\
&= \int_{ \beta = \normb }^{\infty}
\beta \rho_{\mathcal A_i}(\beta)
\;\d \beta\nonumber \\
&\leq 
\sum_{t=2}^{\infty}
2^{t+1} \anormb
\Pr_{z \sim \mathcal A_i} \lp[ \mathcal G(z,t) \rp] \nonumber \\
&\leq 
\sum_{t=2}^{\infty}
2^{t+1} \anormb
\Pr_{z \sim \mathcal D_i} \lp[ \mathcal G(z,t) | b_i = 1\rp] \nonumber \\
&\leq 
\sum_{t=2}^{\infty}
2^{t+1} \anormb
\Pr_{z \sim \mathcal D_i} \lp[   b_i = 1 | \mathcal G(z,t)\rp] 
\frac{ \Pr_{z \sim \mathcal D_i}[ \mathcal G(z,t) ] }{  \Pr[ b_i = 1 ] } \nonumber \\
&\leq 
\frac{ n }{ \tau }
\sum_{t=2}^{\infty}
2^{t+1} \anormb
\Pr_{z \sim \mathcal D_i} \lp[   b_i = 1 | \mathcal G(z,t)\rp] 
\Pr_{z \sim \mathcal D_i} [\mathcal G(z,t)] \nonumber \\
&\leq 
\frac{ n }{ \tau }
\sum_{t=2}^{\infty}
2^{t+1} \anormb
\exp\lp( - 2^t \log (nd/\tau) \rp) \;,
\label{eq:algebra-bound}
\end{align}
where the first equality expands the definition of expectation, 
in the third line we bound $\|  z - \tilde \mu \|_2^2 $ from above by $ 2^{t+1} d \log n$ when $z$ lies in the interval $[2^t d \log n, 2^{t+1} d \log n ]$,
in the fourth line we use the fact that $\mathcal A_i$ is the distribution of $D_i$ conditioned on $b_i = 1$,
in the fifth line we apply Bayes' theorem, 
in the sixth line we use \Cref{eq:no-rare} and the fact that we pick some $i$ such that $b_i = 1$, 
in the seventh line we bound the probability $\Pr_{z \sim \mathcal D_i} [\mathcal G(z,t)]$ by $1$ and the conditional acceptance probability 
$\Pr_{z \sim \mathcal D_i} \lp[   b_i = 1 | \mathcal G(z,t)\rp] $ by 
$\exp\lp( - 2^t \log (nd/\tau) \rp)$.

Now we focus on a single term in the summation.
We have that
\begin{align*}
&\frac{ n }{ \tau }
2^{t+1} \anormb
\exp\lp( - 2^t \log (nd/\tau) \rp) \\
&\leq
\exp \lp(  - 2^t \log (nd/\tau) + t+1 + \log \lp( \anormb \rp)+
 \log n + \log(1/\tau) \rp) \\
&\leq 
\exp \lp(  - 2^{t-1} \log (nd/\tau) \rp) \, ,
\end{align*}
where the last inequality is true as long as 
$t \geq 2$ and $\log(nd/\tau)$ is at least some sufficiently large constant.
This shows that \Cref{eq:algebra-bound} can be bounded from above by some geometric series with common ratio 
less than $1/2$ and initial value less than 
$ \exp\lp( - 2 \log(nd/\tau) \rp) \leq \tau^2 / (n^2d^2) $.
This immediately implies that 
$$
\E_{ z \sim \mathcal A_i }    
\lp[  \mathbbm 1 \{ \| z - \meanini \|_2^2 \geq \normb  \}  \| z - \meanini \|_2^2 \rp]
\leq O(\tau^2 / (n^2d^2)) = o(1/n).
$$
This concludes the proof of \Cref{lem:expected-norm-tail}.
\end{proof}

Next we turn to the third property inside the definition of $\Snice$, which states that the second moment matrices $\E[ \lp( x_i  - \tilde \mu \rp) \lp( x_i  - \tilde \mu \rp)^\top  ]$ 
for the accepted samples are bounded from below by $\vec I / 3$.
\begin{lemma}[Eigenvalue Lower Bound on Second Moment Matrix]
\label{lem:eigenvalue-lb}
Let $x_1, \cdots, x_{n}$ be independent random variables in $\R^d$, 
where the $i$-th sample follows the distribution of $\normal(  \mu, \vec {\Sigma}_i )$ for some $\vec \Sigma_i \succeq \vec I / 2$,  $\meanini \in \R^d$ be a fixed vector, and
 $b_i \in \{0,1\} \sim \mathrm{Ber}(\exp( - \| x_i - \meanini \|_2^2 / d ))$ for $i \in [n]$.
Then it holds that
\begin{align*}
\E_{ x_i \sim \normal(  \mu, \vec {\Sigma}_i ) }
\lp[ \lp( x_i - \meanini \rp) \lp( x_i - \meanini \rp)^\top 
\vline b_i = 1
\rp] \succeq \vec I / 3 \; \forall i \in [n].
\end{align*}
\end{lemma}
\begin{proof}
By \Cref{lem:product}, the distribution of $x_i
\sim \normal( \mu, \vec {\Sigma}_i )$ conditioned on $b_i = 1$ is exactly
$\normal( \tilde \mu_i, \tilde {\vec \Sigma}_i )$ for
\begin{align*}
    \vec {\tilde \Sigma}_i = (\vec \Sigma_i^{-1} + 2 \vec I/d)^{-1} \quad \text{and} \quad
        \tilde \mu_i =   \vec {\tilde \Sigma}_i
         \left(   2 \meanini/ d + \vec \Sigma_i^{-1} \mu \right) \;.
\end{align*}
Since we assume that $\vec \Sigma_i \succeq \vec I / 2$, it follows immediately that 
$
\vec {\tilde \Sigma}_i \succeq \vec I / 3.
$
The lemma then follows by noting that
\begin{align*}
\E_{z \sim \normal( \tilde \mu_i, \tilde {\vec \Sigma}_i )}[ ( z - \meanini  ) ( z - \meanini  )^T ]
\succeq
\Cov_{z \sim \normal( \tilde \mu_i, \tilde {\vec \Sigma}_i)} \lp[  z\rp]
= \vec {\tilde \Sigma}_i.
\end{align*}
This concludes the proof of \Cref{lem:eigenvalue-lb}.
\end{proof}

Lastly, we show the final property of the acceptance set.
Namely, we will show that when at least $\alpha$-fraction of the $n$ samples are drawn from distributions with covariances $\vec \Sigma_i \preceq \vec I$, then the acceptance set must have size at least $\Omega(\alpha n)$.
\begin{lemma}[Acceptance Set Size]
\label{lem:aceptance-set-size}
Let $x_1, \cdots, x_n $ be independent random vectors in $\R^d$, where $x_i$ has mean $\mean \in \R^d$ and covariance $\vec \Sigma_i \in \R^{d \times d}$, and $\meanini \in \R^d$ be a fixed vector satisfying $\| \meanini - \mu \|_2 \lesssim \sqrt{d}$.
Assume that for at least $\alpha n$ many points we have $\vec \Sigma_i \preceq \vec I$ and $\alpha n \gg \log(1/\tau)$ for some $\tau \in (0, 1)$.
Let $b_i \in \{0,1\}$ for $i \in [n]$ be the drawn from $\mathrm{Ber}(\exp( - \| x_i - \meanini \|_2^2 / d ))$. 
Then it holds that $\sum_{i \in [n]} b_i \gtrsim \alpha n$ with probability at least $1 - \tau$.
\end{lemma}
\begin{proof}
The Chernoff bound states that 
$\sum_{i \in [n]} b_i \gtrsim \E[\sum_{i \in [n]} b_i ]$ with probability at least $1 - \tau$, provided that $\E[\sum_{i \in [n]} b_i ] \gg \log(1/\tau)$.
Hence, it suffices to show that 
$ \sum_i \E[b_i] \geq \Omega(\alpha n)$ since we have assumed that $\alpha n \gg \log(1/\tau)$.

Consider a sample for which $\vec \Sigma_i \preceq \vec I$.
It then follows that 
$$\E \lp[ \| x_i - \meanini \|_2^2 \rp] 
\leq 2 \E \lp[ \| x_i - \mu \|_2^2 \rp] + 2   \| \mu - \meanini \|_2^2   = 2\tr \lp( 
  \vec \Sigma_i     \rp) + 2   \| \mu - \meanini \|_2^2  \lesssim d \, ,
$$
where the first inequality follows from the triangle inequality and linearity of expectation, the last inequality follows from the assumption that $\vec \Sigma_i \preceq \vec I$ and $\| \meanini - \mu \|_2 \lesssim \sqrt{d}$.
  
Using Markov's inequality, we have that 
$ \Pr[ \| x_i - \meanini \|_2^2 \geq C d ] \leq \E \lp[ \| x_i - \meanini \|_2^2 \rp]/(C d) \leq 1/2   $
for $C$ being a large enough constant.
Therefore, it follows that 
\begin{align*}
\Pr[ b_i = 1 ]
\geq \Pr \left[ b_i = 1 \; \mid \; \| x_i -\meanini\|_2^2 \leq C d \right]
\Pr[ \|x_i - \meanini\|_2^2 \leq C d ]
\geq \frac{e^{-C}}{2}.
\end{align*}
Since we assume that there are at least $\alpha$ fraction of samples whose covariances are bounded from above by $\vec I$, we immediately have that
$$
\sum_{i=1}^n \E[ b_i = 1 ] \geq  \frac{\alpha n}{2} e^{-C}  = \Omega(\alpha n).
$$
This concludes the proof of \Cref{lem:aceptance-set-size}.
\end{proof}
Combining \Cref{lem:norm_calibaration,lem:expected-norm-tail,lem:eigenvalue-lb,lem:aceptance-set-size} then allows us to conclude that the rejection sampling yields a nice set following the definition of \Cref{eq:event_calibration} with high probability.
\begin{corollary}[Niceness of $S_{\accept}$]
\label{cor:S-niceness}
Let $x_1, \cdots, x_n \in \R^d$ be independent random vectors satisfying the data generation model of \Cref{def:model2} with a common mean $\mean$ and signal-to-noise rate $\alpha$, and $\meanini$ be some vector satisfying $\| \meanini - \mu \|_2 \lesssim \sqrt{d}$.
Assume that $n \gg \tfrac{d}{\alpha} \log(\tfrac{d}{\tau})\log(\tfrac{d}{\alpha \tau})$.
Let $S_{\accept} \subseteq [n]$ be generated by including the index $i \in [n]$ with probability 
$\exp( - \| x_i - \meanini\|_2^2 /d  )$. 
Then, with probability at least $1-\tau$, it holds that $S_{\accept} \in \Snice$ for $\Snice$ defined in \Cref{eq:event_calibration}.
\end{corollary}
\begin{proof}
The first two conditions of \Cref{eq:event_calibration} are satisfied due to \Cref{lem:norm_calibaration,lem:expected-norm-tail} respectively.
Due to \Cref{remark:cov-lb}, we can assume that $\vec \Sigma_i \succeq \vec I / 2$ for all $i \in [n]$.
Applying \Cref{lem:eigenvalue-lb} then allows us to conclude the third condition.

Lastly, assuming that $\alpha n \gg d \log(nd/\tau) \log(d/\tau)$, 
\Cref{lem:aceptance-set-size} implies that the size of the accepted set will be $|S_\cA| = \Omega(\alpha n)$, which then implies that $|S_\cA| \geq 
    \max \lp( C d \log(nd/\tau) \log(d/\tau), \alpha n / C\rp)$ 
    for some sufficiently large constant $C$. This shows that the last condition of \Cref{eq:event_calibration} is satisfied whenever $\alpha n \gg d \log(nd/\tau) \log(d/\tau)$.

    The proof is completed by finally noting that the assumption in the corollary statement $n \alpha \gg d \log(\tfrac{d}{\tau})\log(\tfrac{d}{\alpha \tau})$ implies the condition $\alpha n \gg d \log(nd/\tau)\log(d/\tau)$ mentioned in the previous paragraph. This can be seen using the elementary \Cref{fact:implicitinequality}. This concludes the proof of \Cref{cor:S-niceness}.
\end{proof}

\subsection{Identifying a Low-Variance Subspace}
\label{sec:search-subspace}
In this subsection, we show that with high probability the routine \SearchSubSpace($\cdot$) from \Cref{alg:recursive} will successfully identify a subspace $\cV$ such that the variance of the accepted samples along any direction within the subspace is low.
By \Cref{cor:S-niceness}, if we perform rejection sampling centered at some vector $\meanini$ satisfying $\| \meanini - \mu \|_2 \lesssim \sqrt{d}$, the accepted set will satisfy the nice properties defined in \Cref{eq:event_calibration}.
Conditioned on that, we can view the input samples $x_1, \cdots, x_k$ to \SearchSubSpace($\cdot$) as if they are independent samples from the distributions $\mathcal A_1, \cdots, \mathcal A_k$ satisfying the properties from \Cref{eq:event_calibration} rephrased below:
\begin{align} 
\label{eq:rare-large-norm}
    &\sum_{i \in [k]} 
    \Pr_{z \sim \mathcal A_i}
    \lp[ \| z - \meanini \|_2^2 \geq R \rp]
    \leq \tau \, ,  \\
\label{eq:norm-small-tail}
    &\E_{z \sim \mathcal A_i}
    \lp[ 
    \mathbbm 1 \{ z \geq R \}
    \| z - \meanini  \|_2^2
    \rp]
    \leq o(1/k)  \, , \\ 
\label{eq:eigen-lb}
    &\E_{z \sim \mathcal A_i}
    \lp[ 
    \lp( z - \meanini  \rp)
    \lp( z - \meanini  \rp)^\top
    \rp]
    \succeq \vec I / 3 \, ,
\end{align}
for some number $R>0$ and $\tau \in (0, 1)$ (the specific value for $R$ from \Cref{eq:event_calibration} is not important for the results of this section).
Given such a sequence of independent vectors, we show that \SearchSubSpace($\cdot$) succeeds with probability at least $1 - O(\tau)$ in identifying a $(d/2)$-dimensional subspace with variance at most $O(R / d)$ \footnote{We assume that $d$ is even for convenience.}.
\begin{lemma}[Low-Variance Subspace Identification]
\label{lem:low-var-identification}
Let $x_1 , \cdots, x_k \in \R^d$ be independent random vectors that follow the distributions $\{\mathcal A_i\}_{i=1}^k$ satisfying \Cref{eq:rare-large-norm,eq:norm-small-tail,eq:eigen-lb}.
Further assume that $k \gg R \log(d/\tau)$ and $d$ is an even integer.
Then an execution of \SearchSubSpace({$\meanini, x_1, \cdots x_k$}) from \Cref{alg:recursive} produces a subspace\footnote{Technically, 
the routine produces row orthonormal matrices $\vec P_{\text{high}},\vec P_{\text{low}} \in \R^{d/2 \times d} $.
The $\cV$ referred to here is the subspace spanned by the rows of $\vec P_{\text{low}}$.
} $\cV \subset \R^d$ {of dimension $d/2$} such that the following holds with probability at least $1 - \tau$:
 for all unit vectors $v \in \cV$
$$
v^\top \frac{1}{k} \sum_{i \in [k]} 
\tilde {\vec \Sigma}_i v
\lesssim  R/d \;,
$$
where $\tilde {\vec \Sigma}_i = \Cov_{z \sim \mathcal A_i}[ z ]$.
\end{lemma}
As the first step, we show that the empirical second moment matrix (centered at $\tilde \mu$) of the accepted samples has bounded trace. 
The lemma follows immediately from \Cref{eq:rare-large-norm}.
\begin{lemma}\label{lem:trace_bound}
Let $x_1, \cdots, x_k \in \R^d$ be independent random variables with $\sum_{i=1}^k \Pr [  \| x_i - \rejmean \|_2^2 > R ] \leq \tau$ for some $R>0$ and $\tau \in (0,1)$.
Then with probability at least $1 - \tau$, it holds that 
$$
\tr\lp( \frac{1}{k} \sum_{i=1}^k \lp( x_i - \rejmean \rp) 
\lp( x_i - \rejmean \rp)^\top 
\rp)
\leq R \;.
$$
\end{lemma}
\begin{proof}
Let $\cE$ be the event $\{ \exists i \in [k]: \| x_i - \meanini \|_2^2 > R\}$. By the union bound and the assumption $\sum_{i=1}^k \Pr\lp[ 
\| x_i - \meanini \|_2^2 > R \rp] \leq \tau$, we have that $\Pr[\cE] \leq \tau$. Under the complement of $\cE$, we have that
\begin{align*}
    \tr\lp( \frac{1}{k} \sum_{i=1}^k \lp( x_i - \meanini \rp) 
\lp( x_i - \meanini \rp)^\top 
\rp) =  \frac{1}{k}\sum_{i=1}^k \| x_i - \meanini \|_2^2 \leq R\;.
\end{align*}
This concludes the proof of \Cref{lem:trace_bound}.
\end{proof}
Given that the trace of the empirical second moment matrix is small, it follows that there exist at least $d/2$ many eigenvectors whose corresponding eigenvalues are at most $O(R/d)$. Consequently, these eigenvectors span a subspace with low \emph{empirical variance}.
\begin{lemma}[Subspace with Low Empirical Variance]
\label{cor:subspace-existence}
Let $x_1, \cdots, x_k \in \R^d$ be independent random variables with $\sum_{i=1}^k \Pr [  \| x_i - \rejmean \|_2^2 > R ] \leq \tau$ for some $R>0$ and $\tau \in (0,1)$.
Assume that $d$ is an even integer.
Define 
$$
\Memp := \frac{1}{k} \sum_{i \in [k]} (x_i - \meanini) (x_i - \meanini)^\top.
$$
Let $\cV$ be the subspace spanned by the $d/2$ eigenvectors corresponding to the smallest $d/2$ eigenvalues of $\Memp$, 
and $\vec \Pi_{\cV}$ be the orthogonal projection matrix onto $\cV$.
With probability at least $1 - \tau$, it holds that
$$
\vec \Pi_{\cV} \Memp \vec \Pi_{\cV}^\top
\preceq O(R/d)  \; \vec I.
$$
\end{lemma}
\begin{proof}
Suppose that we have the spectral decomposition $\frac{1}{k} \sum_{i=1}^k \lp(x_i - \meanini\rp) \lp(x_i - \meanini\rp)^\top = \sum_{i=1}^d \lambda_i v_i v_i^\top$, where $v_i$ are the eigenvectors and $\lambda_i$ are the corresponding eigenvalues 
with $\lambda_1 \geq \lambda_2 \geq \cdots \geq \lambda_d$.
Assume for the sake of contradiction that $\lambda_{d/2+1} > 2 R/d$.
Then we would immediately have that
\begin{align}
\label{eq:contradict-norm-bound}
\sum_{i \in [n]} \lambda_i \geq \sum_{i= d/2 + 1}^d \lambda_i
> 2 \frac{R}{d}  \frac{d}{2} = R \;.
\end{align}
Yet, by \Cref{lem:trace_bound}, we have that 
$ \sum_{i \in [n]} \lambda_i = \tr \lp( \frac{1}{k} \sum_{i\in [k]} \lp(x_i - \meanini \rp) \lp(x_i - \meanini \rp)^\top \rp) \leq R$, contradicting \Cref{eq:contradict-norm-bound}.
Hence, it follows that $\lambda_{d} \leq \cdots \leq \lambda_{d/2+1} \leq 2 R/d $. Since $\cV$ has been defined to be the subspace spanned by $v_{d}, \cdots, v_{d/2+1}$ the proof of \Cref{cor:subspace-existence} is complete. 
\end{proof}

The above shows that there exists a subspace $\cV$
such that the eigenvalues of the \emph{empirical} moment matrix 
$ \frac{1}{k} \sum_{i=1}^k \lp( x_i - \meanini \rp) \lp( x_i - \meanini \rp)^\top $
are small.
However, in \Cref{lem:low-var-identification} we would like to show a similar bound but for the \emph{population} covariance matrix $\tfrac{1}{k}\sum_{i=1}^k  \tilde {\vec \Sigma}_i$. 
It turns out that one can show that 
the population covariance matrix is  bounded from above by the former empirical moment matrix in Loewner order with high probability. 
The formal statement is given below.
\begin{lemma}[Covariance Concentration]
\label{lem:cov-concentration}
Let $x_1, \cdots, x_k \in \R^d$ be independent random variables that follow the distributions $\{ \mathcal A_i \}_{i=1}^k$ satisfying \Cref{eq:rare-large-norm,eq:norm-small-tail,eq:eigen-lb}.
Further assume that $k \gg R \log(d/\tau)$.
Then with probability at least $1-2\tau$ it holds that
$$
v^\top \left( \frac{1}{k} \sum_{i=1}^k 
\tilde {\vec \Sigma}_i \right) v 
\leq 1.2 \, v^\top \left( \frac{1}{k} \sum_{i=1}^k ( x_i - \meanini )
( x_i - \meanini )^\top \right) v
$$
for every $v\in \R^d$.
\end{lemma}
To relate the two quantities, we will show that (i) the population covariance matrix is always bounded from above by 
$ \frac{1}{k} \sum_{i=1}^k 
\E \lp[ \lp( x_i - \meanini \rp) \lp( x_i - \meanini \rp)^\top \rp] $ in Loewner order, 
and (ii) the empirical moment matrix 
$ \frac{1}{k} \sum_{i=1}^k \lp( x_i - \meanini \rp) \lp( x_i - \meanini \rp)^\top $
concentrates around its expected value in spectral norm with high probability.
The argument is formalized in the following analysis.

\begin{proof}[Proof of \Cref{lem:cov-concentration}]
Let $\tilde \mu_i$ and $\tilde {\vec \Sigma}_i$ be the mean and covariance of $x_i$ respectively, and $\meanini$ be the fixed vector appearing in \Cref{eq:rare-large-norm,eq:norm-small-tail,eq:eigen-lb}.
Then we have that
\begin{align}
&\E \lp[ 
\lp( x_i - \meanini \rp)
\lp( x_i - \meanini \rp)^\top
\rp]     \nonumber \\
&= 
\E \lp[ 
\lp( x_i - \tilde \mu_i + \tilde \mu_i - \meanini \rp)
\lp( x_i - \tilde \mu_i + \tilde \mu_i - \meanini \rp)^\top
\rp]  \nonumber \\
&= 
\E \lp[ 
\lp( x_i - \tilde \mu_i  \rp)
\lp( x_i - \tilde \mu_i  \rp)^\top
\rp]  
+ 
\E \lp[ 
\lp( x_i - \tilde \mu_i  \rp)
\rp]  \lp( \tilde \mu_i - \meanini \rp)^\top
+ 
\lp( \tilde \mu_i - \meanini \rp) \E \lp[ 
\lp( x_i - \tilde \mu_i  \rp)^\top
\rp]
+ \lp( \tilde \mu_i - \meanini \rp) \lp( \tilde \mu_i - \meanini \rp)^\top 
\nonumber \\
&= 
\tilde {\vec \Sigma}_i
+ \lp( \tilde \mu_i - \meanini \rp) \lp( \tilde \mu_i - \meanini \rp)^\top
\succeq \tilde {\vec \Sigma}_i \, ,
\label{eq:var-min-pt}
\end{align}
where in the second equality we use the definition of $\tilde \mu_i$ and $\tilde {\vec \Sigma}_i$, and in the last inequality we use the fact that
$\lp( \tilde \mu_i - \meanini \rp) \lp( \tilde \mu_i - \meanini \rp)^\top$ is a PSD matrix.
For convenience, we use the following notation for the population second moment matrix, centered around $\meanini$:
\begin{align}
\label{eq:def-Mpop}
\Mpop := \frac{1}{k}
\sum_{i=1}^k
\E[ (x_i - \meanini) (x_i - \meanini)^\top ].    
\end{align}
It follows immediately from \Cref{eq:var-min-pt} that it suffices for us to show that with probability at least $1 - \tau$
\begin{align}
\label{eq:op-norm-mul}
 v^\top \Mpop  v
\leq 1.2  v^\top  \lp( \frac{1}{k} \sum_{i=1}^k 
\lp( x_i - \meanini \rp) \lp( x_i - \meanini \rp)^\top \rp)
v  
\end{align}
for all $v \in \R^d$.

By assumption \Cref{eq:eigen-lb}, we have that each $\E \lp[ (x_i - \tilde \mu)( x_i - \tilde \mu )^\top \rp]$ is bounded from below by $\vec I / 3$.
It follows immediately that
$ \Mpop \succeq \vec I / 3.$
Hence, $\Mpop^{-1/2}$ is well defined and satisfies that
\begin{align}
\label{eq:Mpop-lb}
\|\Mpop^{-1/2}  \| = O(1).
\end{align}
Since our goal is to show a multiplicative bound on the concentration of $\Mpop$ in every direction, as stated in \Cref{eq:op-norm-mul}, we can instead consider the \emph{whitened} version of \Cref{eq:op-norm-mul}. 
That is, it suffices to show that, with probability at least $1 - \tau$, the following holds for all  $u \in \R^d$:
\begin{align}
\label{eq:op-norm-mul-whitened}
 u^\top   u
\leq 1.2 
u^\top  \Mpop^{-1/2}\lp( \frac{1}{k} \sum_{i=1}^k 
\lp( x_i - \meanini \rp) \lp( x_i - \meanini \rp)^\top \rp)
\Mpop^{-1/2}
u.        
\end{align}
It is not hard to see that \Cref{eq:op-norm-mul-whitened} immediately implies \Cref{eq:op-norm-mul} simply by setting $u = \Mpop^{1/2} v$.

Consider the random matrices defined as
$$
\vec X_i = 
\Mpop^{-1/2}
(x_i - \meanini) (x_i - \meanini)^\top 
\Mpop^{-1/2}.
$$
By construction, we have that
\begin{align}
\label{eq:Xi-sum}
\E\left[  \frac{1}{k}
\sum_{i=1}^k \vec X_i \right]
= \Mpop^{-1/2} \Mpop \Mpop^{-1/2} = \vec I.
\end{align}
Ideally, we would like to apply the Matrix Bernstein Inequality on $\{ \vec X_i \}_{i=1}^k$ directly to show that $\tfrac{1}{k}\sum_{i \in [k]} \vec X_i 
\succeq  \vec I / 1.2 $ (which would imply \eqref{eq:op-norm-mul-whitened}).
Yet, note that $\vec X_i$ has neither bounded operator norm nor 
is the expectation of $\vec X_i$ equal to $\vec 0$.
To circumvent the issues, we will instead define the random variables
\begin{align*}
g_i = \mathbbm 1 \{ \| x_i - \meanini \|_2^2 \leq R \} \, ,
\vec Y_i = 
g_i \vec X_i  \, , \,
\bar {\vec Y}_i = 
g_i \vec X_i - \E[g_i \vec X_i] \, , \,
\vec Z_i = 
(1 - g_i) \vec X_i.
\end{align*}
We will then analyze the concentration properties of the above ``truncated'' random matrices, and relate them to $\vec X_i$.

We first analyze $\vec Z_i$.
By the union bound and assumption \Cref{eq:rare-large-norm}, we have that
$$ 
\Pr[ \exists i \in [k]: \vec Z_i \neq 0 ] \leq
\sum_{i=1}^k
\Pr[ \vec Z_i \neq 0 ] 
= 
\sum_{i=1}^k \Pr
\lp[ \| x_i- \meanini \|_2^2 \geq R \rp]
\leq \tau.
$$
In other words, we must have
\begin{align}
\label{eq:zero-z}
\vec Z_i = \vec 0 \, \forall i \in [k]
\end{align}
with probability at least $1 - \tau$.

Next we turn our attention to the relationship between $\vec Y_i$ and $\bar {\vec Y}_i$.
By assumption \Cref{eq:norm-small-tail}, we have that
$$
\tr \lp(  \E[ (1 - g_i) \vec X_i ]  \rp)
= \E \lp[  \mathbbm 1 \{ 
\| x_i - \meanini \|_2^2 \geq R
\} \| x_i - \meanini \|_2^2  \rp]
\leq o(1/k).
$$
Thus, it follows that
$$
\frac{1}{k} \sum_{i \in [k]} \E[ g_i \vec X_i ]
\succeq 
\frac{1}{k} \sum_{i \in [k]}  \E[ \vec X_i ]
- o(1) \vec I
= \lp( 1 - o(1) \rp) \vec I \, ,
$$
where the last equality follows from \Cref{eq:Xi-sum}.
Recall that $\bar {\vec Y}_i$ is defined as
$ g_i \vec X_i - \E[g_i \vec X_i] = \vec Y_i  - \E[g_i \vec X_i]$.
Hence, we must have
\begin{align}
\label{eq:Y-barY-relationship}
\frac{1}{k} \sum_{i \in [k]}
\vec Y_i
\succeq
\frac{1}{k} \sum_{i \in [k]}
\bar {\vec Y}_i + ( 1 - o(1) ) \vec I.
\end{align}
Lastly, we analyze the concentration property of $\bar {\vec Y}_i$ by applying the Matrix Bernstein inequality.
In particular, we will show that (i) $\bar {\vec Y}_i$ has mean $\vec 0$, (ii) $\bar {\vec Y}_i$  has its operator norm bounded from above by $R$ almost surely, and (iii) $\E[ \bar {\vec Y}_i^2 ]$ also has its operator norm bounded by $O(R)$.
Claim (i) follows from the definition of $\bar {\vec Y}_i$.
For claim (ii), we note that
$$
\| g_i \vec X_i - \E[ g_i \vec X_i ] \|_2
\leq \max \lp( \| g_i \vec X_i \|_2,  \| \E[ g_i \vec X_i ] \|_2 \rp)
\leq R
$$
by the definition of $g_i$.
We proceed to bound from above the operator norm of $ \bar {\vec Y}_i^2$. 
Note that
\begin{align*}
&\E \lp[ \bar {\vec Y}_i^2 \rp]  = \E \lp[ g_i  
\lp( \Mpop^{-1/2}  \lp( x_i - \meanini \rp)
\lp( x_i - \meanini \rp)^\top
\Mpop^{-1/2} \rp)^2
\rp]
-  \E \lp[ 
g_i  \Mpop^{-1/2}  \lp( x_i - \meanini \rp)
\lp( x_i - \meanini \rp)^\top
\Mpop^{-1/2}
\rp]^2
\\
&\preceq
\E \lp[  
g_i
\Mpop^{-1/2}  \lp( x_i - \meanini \rp)
\lp( x_i - \meanini \rp)^\top
\Mpop^{-1} \lp( x_i - \meanini \rp)
\lp( x_i - \meanini \rp)^\top \Mpop^{-1/2}
\rp] \\
&\preceq 
O(R)
\E_{x_i \sim \normal(\tilde \mu_i, \tilde {\vec \Sigma}_i)} \lp[  
\Mpop^{-1/2}  \lp( x_i - \meanini \rp)
\lp( x_i - \meanini \rp)^\top \Mpop^{-1/2}
\rp] \;,
\end{align*}
where the first inequality is true as $g_i$ is an indicator variable, and so $g_i  \Mpop^{-1/2}  \lp( x_i - \meanini \rp)
\lp( x_i - \meanini \rp)^\top
\Mpop^{-1/2}$ is a PSD matrix, the second inequality is true 
due to \Cref{eq:Mpop-lb}
($\|\Mpop^{-1} \|_2 \leq O(1)$), which implies that 
$g_i \lp( x_i - \meanini \rp)^\top
\Mpop^{-1} \lp( x_i - \meanini \rp) \leq 
\mathbbm 1 \{ \| x_i - \meanini  \|_2^2 \leq R \} 
\| x_i - \meanini  \|_2^2
\|\Mpop^{-1} \|_2
= O(R)$ almost surely.
It then follows that
\begin{align*}
\frac{1}{k} \sum_{i=1}^k \E \lp[ \bar {\vec Y_i}^2 \rp]    
&\preceq O(R)
\Mpop^{-1/2}
\frac{1}{k} \sum_{i=1}^k
\E_{x_i \sim \normal(\tilde \mu_i, \tilde {\vec \Sigma}_i)} \lp[ 
  \lp( x_i - \meanini \rp)
\lp( x_i - \meanini \rp)^\top 
\rp] 
\Mpop^{-1/2}
 \\
&= O(R)
\Mpop^{-1/2}
\Mpop 
\Mpop^{-1/2}
 = O(R) \vec I \, ,
\end{align*}
where in the last line we use the definition of $\Mpop$ given in \Cref{eq:def-Mpop}.

Hence, we can apply the Matrix Bernstein inequality (\Cref{fact:matrixBernstein}) with $t=(1-1/1.1)$, $\nu= O(k R)$ and $L=R$, and obtain that
$$
\Pr\lp[ \lp \|  
\frac{1}{k} \sum_{i=1}^k \bar {\vec Y}_i 
\rp \|_2  \geq 1-\frac{1}{1.1} \rp]
\leq d \exp \lp( - O \left( 
\frac{ k  }{  R      }
\right)\rp) \, .
$$
The right hand side is at most $\tau$ as long as
$k > C R \log(d/\tau)$ for a sufficiently large constant $C$.
In other words, we have that
\begin{align}
\label{eq:barY-concentration}
\lp \| \frac{1}{k} \sum_{i=1}^k \bar {\vec Y}_i \rp \|_2
\leq 1 - 1/1.1    
\end{align}
with probability at least $1 - \tau$.

We are now ready to combine things together to show \Cref{eq:op-norm-mul-whitened}.
Combining \Cref{eq:barY-concentration} and \Cref{eq:Y-barY-relationship} gives that
\begin{align}
\label{eq:Yi-operator-bound}
\frac{1}{k}
\sum_{i \in [k]} \vec Y_i
\succeq \lp( 1 - o(1) - 1 + 1/1.1 \rp) \vec I
\succeq \frac{1}{1.2} \vec I
\end{align}
with probability at least $1 - \tau$.
Note that $\vec X_i = \vec Z_i + \vec Y_i$ by construction.
By the union bound, \Cref{eq:Yi-operator-bound,eq:zero-z} hold at the same time with probability at least $1 - 2\tau$.
Conditioned on them, we thus have
$$
\frac{1}{k} 
\sum_{i \in [k]}
\vec X_i
= 
\frac{1}{k} 
\sum_{i \in [k]}
\vec Y_i \succeq \frac{1}{1.2} \vec I \, ,
$$
which immediately implies \Cref{eq:op-norm-mul-whitened}.
This concludes the proof of \Cref{lem:cov-concentration}.
\end{proof}

We are now ready to conclude the proof of \Cref{lem:low-var-identification}.
\begin{proof}[Proof of \Cref{lem:low-var-identification}]
Combining the assumption \Cref{eq:rare-large-norm} ($\sum_{i=1}^k \Pr \lp[ \| x_i - \tilde \mu \|_2^2 \geq R \rp] \leq \tau $) with  \Cref{cor:subspace-existence} and \Cref{lem:trace_bound} gives that with probability at least $1 - \tau$ the following two inequalities hold:
\begin{align}
\label{eq:low-var-projection}
\vec \Pi_{\cV} \frac{1}{k} \sum_{i=1}^k ( x_i - \meanini )
( x_i - \meanini )^\top \vec \Pi_{\cV}^\top
\preceq O( R/d )  \; \vec I \,  \;,
\end{align}
where $\cV$ is the subspace spanned by the rows of $\RowP_{\mathrm{low}}$ that \SearchSubSpace({$\rejmean, x_1, \cdots, x_k$}) from \Cref{alg:recursive} returns.
On the other hand, combining the assumptions from
\Cref{eq:rare-large-norm,eq:norm-small-tail,eq:eigen-lb} and $k \gg R \log(d/\tau)$
with 
\Cref{lem:cov-concentration} gives that with probability at least $1 - \tau$,
\begin{align}
\label{eq:cov-concentration}
    v^\top \frac{1}{k} \sum_{i=1}^k \tilde {\vec \Sigma}_i v \leq 1.1 v^\top \frac{1}{k} \sum_{i=1}^k ( x_i - \meanini )
( x_i - \meanini )^\top v
\end{align}
for any unit vector $v \in \R^d$.
Combining \Cref{eq:low-var-projection,eq:cov-concentration} then gives that
\begin{align*}
\max_{v\in \cV : \|v\|_2 = 1} v^\top  \frac{1}{k} \sum_{i=1}^k \tilde {\vec \Sigma}_i v  
\lesssim  R/d.
\end{align*}
This concludes the proof of \Cref{lem:low-var-identification}.
\end{proof}

\subsection{Bias and Concentration of Mean within the Low-Variance Subspace}\label{sec:bias}
The rejection sampling procedure from \Cref{line:rejection_sampling} in \Cref{alg:partial-estimate}
creates bias on the expectation of the samples.
Specifically, if the original sample $x_i $ follows the distribution $\normal(\mu, \vec \Sigma_i)$, then by \Cref{lem:product}, the distribution of $x_i$ conditioned on acceptance 
(assuming $\tilde \mu$ is the rejection sampling center) will follow a different Gaussian $\normal( \tilde \mu_i, \tilde {\vec \Sigma}_i)$, where
\begin{align}
\label{eq:new-distr-def}
    \tilde \mu_i = \vec {\tilde \Sigma}_i \lp( 2\meanini / d + \vec \Sigma_i^{-1} \mu \rp) 
    \,\,\text{and}\,\,
     \vec {\tilde \Sigma}_i = \lp( \vec \Sigma_i^{-1} + \tfrac{2}{d}\vec I \rp)^{-1} \;.
\end{align}
Fortunately, we can show that the bias averaged over all surviving samples cannot be too large when we restrict the vectors to the low-variance subspace $\cV$ identified by the algorithm.
As the first step, we derive an explicit formula for the averaged mean vector.
\begin{lemma}[Mean Bias Form]
\label{lem:error_on_low_var}
Let $x_1, \ldots, x_k \in \R^d$ be independent samples following the distributions $\normal( \tilde \mu_i, \tilde {\vec \Sigma}_i)$, where 
$\tilde \mu_i, \tilde {\vec \Sigma}_i$ are functions of $\tilde \mu, \mu
\in \R^d, \vec \Sigma_i \in \R^{d \times d}$  defined as in \Cref{eq:new-distr-def}.
Then it holds that 
$$
\frac{1}{k}
\sum_{i=1}^k \tilde \mu_i
- \mu
=  \frac{2}{d} \left(\frac{1}{k} \sum_{i=1}^k \vec {\tilde \Sigma}_i \right)\lp( 
\meanini - \mu \rp) .
$$
\end{lemma}
\begin{proof}
Using the expressions from \eqref{eq:new-distr-def} and direct manipulations, we can write
\begin{align*}
\frac{1}{k} \sum_{i=1}^k \tilde \mu_i - \mu
&= \frac{1}{k}  \sum_{i=1}^k \tilde {\vec \Sigma}_i\left(  2\tilde \mu/d + \vec \Sigma_i^{-1} \mu - \tilde {\vec \Sigma}_i^{-1} \mu \right) \\
&= { \frac{1}{k}\sum_{i=1}^k
\vec {\tilde \Sigma}_i
\lp( 2\meanini / d + \vec \Sigma_i^{-1} \mu  
- \vec \Sigma_i^{-1} \mu - 2\mu / d \rp) 
} \\
&= 
\frac{2}{k} \sum_{i=1}^k \tilde {\vec \Sigma}_i \frac{\lp( \meanini - \mu \rp)}{d} \;.
\end{align*}
\end{proof}

As a consequence of \Cref{lem:error_on_low_var}, if it were the case that $\|  \frac{1}{k} \sum_{i=1}^k \tilde {\vec \Sigma}_i \| \ll d$, one could show that the $\ell_2$ norm of the bias vector would be significantly smaller than the distance between $\mu$ and the estimate $\tilde \mu$ that we start with.
While the operator norm of $\frac{1}{k} \sum_{i=1}^k \tilde {\vec \Sigma}_i$ can in general be quite large, \Cref{lem:low-var-identification} provides a routine for us to search for {a} subspace where the spectral norm of 
$\frac{1}{k} \sum_{i=1}^k \tilde {\vec \Sigma}_i$ is bounded. Inside that small variance subspace, we will use the 
 lemma below to bound from above the norm of the bias vector.

\begin{lemma}
[Bias Bound within Subspace]
\label{lem:mean-projection-bound}
Let $x_1, \ldots, x_k \in \R^d$ be independent samples following the distributions $\normal( \tilde \mu_i, \tilde {\vec \Sigma}_i)$, where 
$\tilde \mu_i, \tilde {\vec \Sigma}_i$ are functions of $\tilde \mu, \mu
\in \R^d, \vec \Sigma_i \in \R^{d \times d}$  defined as in \Cref{eq:new-distr-def}.
Let $\cV$ be a subspace of $\R^d$ and $\vec \Pi_{\cV} \in \R^{d \times d}$ be the orthogonal projector onto $\cV$.
Then it holds that 
$$
\lp\| 
\vec \Pi_{\cV}
\lp( 
\frac{1}{k}
\sum_{i=1}^k \E \lp[ x_i \rp]
- \mu   \rp) \rp \|_2 \lesssim  \frac{1}{d} \sqrt{\left\| \frac{1}{k}\sum_{i=1}^k  \tilde {\vec \Sigma}_i \right\|_2 \left(\max_{v \in \cV : \|v\|_2=1} v^\top \frac{1}{k} \sum_{i=1}^k \tilde {\vec \Sigma}_i v \right)} \;  \lp \|   \meanini - \mu    \rp \|_2.
$$
\end{lemma}
\begin{proof}
    For convenience
we define $
\Sigmavg = 
\frac{1}{k} \sum_{i=1}^k \tilde {\vec \Sigma}_i 
$.

By \Cref{lem:error_on_low_var}, we have that
\begin{align}
\vec \Pi_{\cV}
\lp( 
\frac{1}{k}
\sum_{i=1}^k \E \lp[ x_i \rp]
- \mu   \rp)
= \frac{2}{d}\vec \Pi_{\cV}
\Sigmavg \lp( \meanini - \mean \rp) . \label{eq:normineq}
\end{align}
{Consider an arbitrary} $v \in \vec \Pi_\cV$. We have that
\begin{align}
\label{eq:bound-via-op-norm}
\left| v^\top \vec \Pi_{\cV}\lp( \frac{1}{k}\sum_{i=1}^k \E \lp[ x_i \rp]- \mu   \rp)  \right|
&= \frac{2}{d} \left|v^\top \vec \Pi_{\cV}\Sigmavg ( \meanini - \mu) \right| \tag{using \Cref{eq:normineq}}\\
&= \frac{2}{d} \left|v^\top  \Sigmavg ( \meanini - \mu) \right| \tag{since $v \in \vec \Pi_\cV$}\\
&= \frac{2}{d} \left| \lp( \Sigmavg^{1/2} v \rp)^\top  \Sigmavg^{1/2} ( \meanini - \mu) \right| \notag \\
&\leq \frac{2}{d} \sqrt{ v^\top \Sigmavg v  }\sqrt{( \meanini - \mu)^\top \Sigmavg ( \meanini - \mu)  }  \tag{Cauchy–Schwarz inequality}\\
&\leq \frac{2}{d} \sqrt{ v^\top \Sigmavg v  } \sqrt{\|\Sigmavg\|_2 }    \| \meanini - \mu \|_2 \;, \label{eq:normbound1}
\end{align}
where the penultimate line follows by the fact that $\Sigmavg$ is PSD matrix combined with the Cauchy-Schwarz inequality.
We conclude the proof by taking the maximum over $v \in \cV$ on both sides:
\begin{align*}
    \lp\| \vec \Pi_{\cV}\lp( \frac{1}{k}\sum_{i=1}^k \E \lp[ x_i \rp]- \mu   \rp) \rp \|_2
    &\leq \max_{v \in \cV : \|v\|_2=1} v^\top \vec \Pi_{\cV}\lp( \frac{1}{k}\sum_{i=1}^k \E \lp[ x_i \rp]- \mu   \rp) \\
    &\leq \frac{2}{d}  \sqrt{ \max_{v \in \cV : \|v\|_2=1} v^\top \Sigmavg v  } \sqrt{\|\Sigmavg\|_2 }  \| \meanini - \mu \|_2 \;,
\end{align*}
where the first line is follows by the variational characterization of $\ell_2$-norm, and the second line follows by \Cref{eq:normbound1}.
\end{proof}

The above concerns the population mean of the accepted samples. We next turn our attention to the empirical mean of the accepted samples. 
Since they are independent Gaussian variables, it follows from standard concentration that the deviation of the empirical mean from its expected value along {any} direction $v$ is proportional to the variance along that direction.

\begin{lemma}[Mean Concentration]
\label{lem:overall-mean-concentration}
Let $x_1, \cdots, x_k \in \R^d$ be independent random variables distributed as $x_i \sim \normal( \tilde{\mu}_i, \vec {\tilde\Sigma}_i )$.
Then, with probability at least $1 - \tau$, it holds that
$$
 v^\top \lp( \frac{1}{k} \sum_{i=1}^k x_i - \frac{1}{k} \sum_{i=1}^k \tilde{\mu}_i \rp)
\lesssim 
\sqrt{\frac{d + \log(1/\tau)}{k}}  \sqrt{ \frac{1}{k} \sum_{i=1}^k  v^\top    \vec {\tilde\Sigma}_i v }\;,
$$
for all unit vectors $v \in \R^d$.
\end{lemma}
\begin{proof}
For convenience, we denote $\Sigmavg := \frac{1}{k} \sum_{i=1}^k   \vec {\tilde\Sigma}_i  $.
The random variable $
\Sigmavg^{-1/2} \frac{1}{k} \sum_{i=1}^k
\lp( x_i - \tilde \mu_i \rp)$ follows the distribution $\normal(0, \vec I / k )$.
Hence,  the standard concentration bound for spherical Gaussian vectors yields that
\begin{align*}
\Pr\lp[
\max_{u \in \R^d : \|u\|_2=1} u^\top \Sigmavg^{-1/2} \frac{1}{k}
\sum_{i=1}^k
 \lp( x_i - \tilde{\mu}_i \rp) 
> C 
\sqrt{  \frac{d + \log(1/\tau)}{k}  }
\rp]
\leq \tau \, ,
\end{align*}
where $C$ is some sufficiently large absolute constant $C$.
Plugging in the unit vector
$u = \Sigmavg^{1/2} v / \sqrt{ v^\top \Sigmavg v }$
then gives that
\begin{align*}
\max_{v \in \R^d : \|v\|_2=1}
v^\top  \frac{1}{k}
\sum_{i=1}^k
 \lp( x_i - \tilde{\mu}_i \rp)
 \leq  C \sqrt{  \frac{d + \log(1/\tau)}{k}  } \; \sqrt{ v^\top \Sigmavg v } 
\end{align*}
with probability $1-\tau$.
This concludes the proof of \Cref{lem:overall-mean-concentration}.
\end{proof}

\subsection{Mean Estimation Improvement}\label{sec:mean_improvement}
We are now ready to prove the main result of this section, demonstrating that \textsc{PartialEstimate} in \Cref{alg:recursive} makes progress on estimating the mean within some low-variance subspace with high probability.
{The result is restated below:}

\ERRORIMPROVEMENT*
\begin{proof}
Let $S_{\accept}$ be the set of indices that the rejection sampling of  \Cref{line:rejection_sampling} of \Cref{alg:recursive} accepts (i.e., $S_{\accept}$ includes all the $i$'s for which $b_i=1$ in \Cref{line:rejection_sampling}).

The vector $\meanini$ is assumed to satisfy $\|\meanini - \mu\|_2 \lesssim \sqrt{d}$. This together with the assumption $n \gg \tfrac{d}{\alpha} \log(\tfrac{d}{\tau})\log(\tfrac{d}{\alpha \tau})$ implies that \Cref{cor:S-niceness} can be applied. The application yields that the set $S_{\accept}$ is nice ($S \in \Snice$ with $\Snice$ defined in \Cref{eq:event_calibration}) with probability at least $1 - \tau$.

Now let us condition on the event $S_{\accept} = S$ for some $S \in \Snice$, and re-label the samples $\{x_i: i \in S\}$ to $\{x_1,\ldots,x_k\}$ (where $k=|S|$). 
These accepted samples are still independent and, by \Cref{lem:product}, the  distribution of $x_i$ conditioned on acceptance is $x_i \sim \cN(\tilde \mu_i, \tilde{\vec \Sigma}_i)$ with $\tilde \mu_i, \tilde{\vec \Sigma}_i$ being a function of $\mu, \meanini, \vec \Sigma_i$ defined as in \Cref{eq:new-distr-def}.
Since $S \in \Snice$, 
we should have that $k \gg d \log(n d/\tau)\log(d/\tau)$ and the distributions $\normal(\tilde \mu_i, \tilde {\vec \Sigma}_i)$ should satisfy the properties from
\Cref{eq:rare-large-norm,eq:norm-small-tail,eq:eigen-lb} with $R = \normb$.
\Cref{lem:low-var-identification} is thus applicable with $R= \normb$ and the following holds with probability at least~$1 - \tau$:
\begin{align}
\label{eq:low-var-success}
\max_{v \in \cV : \|v\|_2=1} v^\top \frac{1}{k} \sum_{i \in [k]} \tilde {\vec \Sigma}_i v \lesssim    \varb  \;,
\end{align}
where $\cV$ is the vector space spanned by the rows of the matrix  $\RowP_{\mathrm{low}}$ produced by \Call{ParitalEstimate}{$\meanini,x_1,\ldots,x_n$}.
We also note that 
\begin{align}\label{eq:spectral-bound}
    \left\| \frac{1}{k}\sum_{i \in [k]} \tilde {\vec \Sigma}_i  \right\|_2 \lesssim d
\end{align}
which follows trivially by the definition of $\tilde {\vec \Sigma}_i$ in \Cref{eq:new-distr-def}.

By \Cref{lem:overall-mean-concentration}, with probability at least $1 - \tau$, it holds that
\begin{align}
\label{eq:mean-concentration-success}    
 v^\top \lp( \frac{1}{k} \sum_{i \in [k]}  x_i - \frac{1}{k} \sum_{i \in [k]} \tilde{\mu}_i \rp)
\lesssim \sqrt{\frac{d + \log(1/\tau)}{k}}  \sqrt{ \frac{1}{k} \sum_{i \in [k]}  v^\top    \vec {\tilde\Sigma}_i v }\;,
\end{align}
for all unit vectors $v \in \R^d$.
By the union bound, \Cref{eq:low-var-success,eq:mean-concentration-success} hold simultaneously with probability at least $1 - 2\tau$.
Pick some $v \in \cV$.
Using \Cref{eq:low-var-success,eq:mean-concentration-success} as well as \Cref{lem:error_on_low_var,lem:mean-projection-bound}, we have that
\begin{align*}
\lp | v^\top \lp(   \mul -   \mu \rp) \rp|
&\leq
\lp| v^\top \lp(  \mul - \frac{1}{k}\sum_{i \in [k]} \tilde \mu_i \rp) \rp|
+  \lp| v^\top \lp( \frac{1}{k}\sum_{i \in [k]} \tilde \mu_i -  \mu \rp) \rp| \tag{the triangle inequality}\\
&= 
\lp| v^\top \lp(  \mul - \frac{1}{k}\sum_{i \in [k]} \tilde \mu_i \rp) \rp|
+ \lp| v^\top \lp( \frac{1}{k} \sum_{i \in [k]} \vec {\tilde \Sigma}_i \rp) \lp( 
\meanini - \mu \rp) \frac{2}{d} \rp| \tag{using \Cref{lem:error_on_low_var}}\\
&\leq 
\lp| v^\top \lp( \mul {-} \frac{1}{k}\sum_{i \in [k]} \tilde \mu_i \rp) \rp|
{+} O\left(\frac{1}{d} \sqrt{\max_{v \in \cV : \|v\|=1} v^\top \frac{1}{k}\sum_{i=1}^k \vec {\tilde \Sigma}_i v  } \sqrt{\left\| \frac{1}{k}\sum_{i=1}^k \vec {\tilde \Sigma}_i\right\|_2}   \right) \|  \meanini - \mean \|_2 \tag{by  \Cref{lem:mean-projection-bound} }\\
&\leq 
\lp| v^\top \lp( \mul {-} \frac{1}{k}\sum_{i \in [k]} \tilde \mu_i \rp) \rp|
{+} O\left(\sqrt{\frac{\varb}{d}}\right) \|  \meanini - \mean \|_2 \tag{by \Cref{eq:low-var-success} and \Cref{eq:spectral-bound}}\\
&\leq 
\lp| v^\top \lp( \mul - \frac{1}{k}\sum_{i \in [k]} \tilde \mu_i \rp) \rp|
+  \frac{1}{\kappa}\|  \meanini - \mean \|_2  \tag{by assumption $d \gg \kappa^2 \varb$}\\
&= 
\lp|   \frac{1}{k}\sum_{i \in [k]} v^\top \lp( x_i - \tilde \mu_i \rp)  \rp|
+  \frac{1}{\kappa} \|  \meanini - \mean \|_2  
\tag{definition of $\mul$ in \Cref{line:mu1} and $v \in \cV$}
\\
&\leq 
O \left( \; \sqrt{ \frac{d+\log(1/\tau)}{k} } \right)
\; \sqrt{ \frac{1}{k} \sum_{i \in [k]} v^\top \tilde {\vec \Sigma}_i v }
+  \frac{1}{\kappa}\|  \meanini - \mean \|_2  \tag{using \eqref{eq:mean-concentration-success}}\\
&\leq 
O\left( \sqrt{\frac{d+\log(1/\tau)}{\alpha n}} \right) \; \sqrt{ \frac{1}{k} \sum_{i \in [k]} v^\top \tilde {\vec \Sigma}_i v }
+  \frac{1}{\kappa}\|  \meanini - \mean \|_2 \tag{$k=\Omega(\alpha n)$ since $S \in \Snice$}\\
&\leq  O\left( \sqrt{\frac{d+\log(1/\tau)}{\alpha n}} \sqrt{\log(nd/\tau)} \right)  +   \frac{1}{\kappa}\|  \meanini - \mean \|_2.
\tag{using \eqref{eq:low-var-success}}
\end{align*}
This concludes the proof of \Cref{lem:low_var_error_impr}.

\end{proof}

\section{Proof of \texorpdfstring{\Cref{thm:main}}{ref{thm:main}}}\label{sec:proof_of_thm}

In this section we combine the lemmata from the previous sections together with an inductive argument to conclude the proof of \Cref{thm:main}. 
We start by analyzing a single iteration of the for loop of \Cref{line:stages} of \Cref{alg:mean_estimation}. 
In particular, the lemma below bounds from above the estimation error of the final output after a full recursive call to \textsc{RecursiveEstimate}.
\begin{lemma}[Accumulation of error terms]\label{lem:single_stage_analysis}
    Let $\delta \in (2,\infty)$ be an absolute constant.
    Let $ D \geq d$ be two integers that are powers of $2$,
    $n \in \Z_+$, $\tau \in (0, 1)$,
    $\alpha \in ( 0 , 1)$,
    $\hat \mu \in \R^d$,
    and $\RowP \in \R^{d \times D}$ be a row orthonormal matrix.
    Assume that $n \gg \tfrac{d}{\alpha} \log(\tfrac{d}{\tau})\log(\tfrac{d}{\alpha \tau})$.
    Given access to batches of samples generated by the model of \Cref{def:model2} with common mean $\mu \in \R^D$ and signal-to-noise ratio $\alpha \in (0,1)$,
    an execution of \Call{RecursiveEstimate}{$\RowP,n,\hat \mu,\tau$} produces 
    some vector $\hat \mu' \in \R^d$ such that
    the following holds with probability at least $1 - \log_2(d) \; (\tau + O(n^{2-\delta}))$:
    \begin{align*}
        \left\|  \hat \mu' - \RowP \mu \right\|_2 &\leq \frac{\log_2 d}{ 10 \L} \left\| \hat{\mu} - \RowP \mu \right\|_2 
        +  \frac{ \log_2 d }{ 10 \log D } \f(\alpha, n)
        \footnotemark
        \\
        &+ O\left( \sqrt{\frac{(d+\log(1/\tau))\log(nd/\tau)}{\alpha n}}\log_2 d   \right) + O(  \log D \; \varb ) \; \f(\alpha,n)  \, ,
    \end{align*}
    \footnotetext{
    We note {this term is smaller than the last term}. We keep the term only to make the inductive argument more explicit.}
    where $\f(\cdot)$ denotes the 1-d estimation error function defined in \Cref{eq:function_f}.
\end{lemma}
\begin{proof}
    We show that the claim holds for $d$ being any power of two by induction.  The base case of the induction is when either (i) $d \leq C (\L)^2 \varb)$, i.e., the condition in \Cref{line:base} of \Cref{alg:recursive}, where $C$ is some sufficiently large constant, or (ii) $\sqrt{d} \leq \f(\alpha, n)$, i.e., the condition in \Cref{line:base2}.
    In~case~(i), \Call{RecursiveEstimate}{$\RowP,n,\hat \mu,\tau$} just runs the base estimator and terminates with its output. The error in that case is $O(\; \f(\alpha,n) \sqrt{d})
    \leq O( \f(\alpha,n) \cdot  \L \cdot \varb)$ with probability at least $1 - \tau$ {by \Cref{cor:naive_multivariate}}.
    In case~(ii), \Call{RecursiveEstimate}{$\RowP,n,\hat \mu,\tau$} runs \textsc{TournamentImprove}. 
    By \Cref{lem:guess-center}, the procedure produces some estimate $\tilde \mu \in \R^d$ such that 
    $\| \tilde \mu - \vec P \mu\|_2 \lesssim \sqrt{d} + \f(\alpha, n)$ with probability at least $1 - \tau$.\footnote{We  use $\RowP \mu$ in place of the $\mu$ that appears in the statement of \Cref{lem:low_var_error_impr}. This is because $x_i$ are the samples after projection using $\RowP$.} 
    Since in this case we have $\sqrt{d} \leq \f(\alpha, n)$, the right hand side is further upper bounded by $2 \f(\alpha, n)$.
    Hence, the lemma is satisfied in both base cases.

    For the inductive step, we assume that the claim holds up to some dimension $d$ and we will show that the same claim holds when $\RowP \in \R^{2d \times D}$.
    Since we are not in the base case, we can assume that
    $\f(\alpha,n) \leq \sqrt{d}$ and 
    $d \geq C (\L)^2 \varb$.
    Observe that with probability at least $1- O(n^{2-\delta})$, the estimator $\meanini$ computed in \Cref{line:tournament} satisfies 
    \begin{align}\label{eq:center_good}
        \| \meanini -  \RowP \mu \|_2 \lesssim \min\left( \|\hat \mu - \RowP \mu\|_2,  \sqrt{d} \right) + \f(\alpha,n) \lesssim  \min\left( \|\hat \mu - \RowP \mu\|_2 + \f(\alpha,n),  \sqrt{d} \right) \;,
    \end{align}
    where the first inequality follows from \Cref{lem:guess-center}, and the second inequality follows from the fact that $\f(\alpha, n) \leq \sqrt{d}$.

    Let $\RowP_{\mathrm{low}},\RowP_{\mathrm{high}}$, and $\mul$ be defined respectively as in \Cref{line:search_subspace,line:mu1} of \Cref{alg:recursive}. 
    In \Cref{line:recursive-call}, the estimator $\muh$ is set to be the output of the recursive call  \Call{RecursiveEstimate}{$\RowP_{\mathrm{high}}\cdot \RowP , n, \RowP_{\mathrm{high}} \hat \mu, \tau$}. The row orthonormal matrix $\RowP_{\mathrm{high}}\cdot \RowP$ used in that call has half the rows of $\RowP$. Thus, the inductive hypothesis is applicable for the estimator $\muh$.
    This implies that with probability at least $1- (\tau + O(n^{2-\delta})) \log_2 d$ the following holds:
    \begin{align}\label{eq:error1}
        &\left\| \muh - \RowP_{\mathrm{high}} \vec P \mu  \right\|_2  
        \leq 
        \frac{\log_2 d}{ 10 \L} \left\| \hat{\mu} - \RowP \mu \right\|_2 \nonumber 
        +   \frac{ \log_2 d }{ 10 \log D } \f(\alpha, n)   \\ 
        &+ O\left( \sqrt{\frac{(d+\log(1/\tau))\varb}{\alpha n}}\log_2 d   \right)
        + O( \L \; \varb ) \; \f(\alpha,n).
    \end{align}
    We now apply \Cref{lem:low_var_error_impr} with $\kappa = c \L$ for some constant $c$ that will be specified later.
    The lemma is applicable since: (i) the condition $\| \meanini -  \RowP \mu \|_2 \lesssim \sqrt{d}$ shown in \Cref{eq:center_good}, (ii) our assumption  $n \gg \tfrac{d}{\alpha} \log(\tfrac{d}{\tau})\log(\tfrac{d}{\alpha \tau})$, and (iii) the fact that $d \geq C (\L)^2 \; \log(n / \tau) \gg \kappa^2 \varb
    = c^2 (\L)^2 \varb$ as long as $C$ is sufficiently large with respect to $c$.
    By the application of \Cref{lem:low_var_error_impr}, we have the following with probability at least $1-\tau$:
    \begin{align}
        \lp \|  \mul -  \RowP_{\mathrm{low}}^\top \RowP_{\mathrm{low}} \RowP \mu \rp \|_2
&\leq \frac{1}{c \L}  \lp \|     \meanini - \RowP \mu  \rp\|_2
+ O\left( \sqrt{\frac{(d + \log(1/\tau))\varb}{\alpha n}} \right) \nonumber \tag{using \Cref{lem:low_var_error_impr}} \\
&\leq O \lp( \frac{1}{c \L} \rp)  \lp \|     \hat \mu - \RowP \mu  \rp\|_2
+ O \lp( \frac{\f(\alpha,n)}{c\L} \rp) + O\left( \sqrt{\frac{(d + \log(1/\tau))\varb}{\alpha n}} \right) 
\nonumber \tag{using \Cref{eq:center_good}} \\
&\leq \frac{1}{10 \L}   \lp \|     \hat \mu - \RowP \mu  \rp\|_2
+  \frac{\f(\alpha,n)}{10 \L} + O\left( \sqrt{\frac{(d + \log(1/\tau))\varb}{\alpha n}} \right) \, ,
\label{eq:error2}
\end{align}
where the last inequality is true as long as $c$ is sufficiently large to cancel out the hidden constant inside the $O(\cdot)$ notation.

By the discussion so far and a union bound, we have that with probability at least $1- \tau- O(n^{2-\delta})  -(\tau + O(n^{2-\delta})) \log_2 d = 1-(\tau + O(n^{2-\delta}))\log_2(2d)$,  \eqref{eq:error1} and \eqref{eq:error2} hold simultaneously.
Note that $\left\| \RowP_{\mathrm{high}}^\top \muh - \RowP_{\mathrm{high}}^\top \RowP_{\mathrm{high}} \RowP\mu \right\|_2 \leq \|  \muh-   \RowP_{\mathrm{high}} \RowP\mu \|_2$ since $\RowP_{\mathrm{high}}$ is orthonormal. 
Combining this observation, \Cref{eq:error1,eq:error2}, and the triangle inequality gives that
\begin{align*}
\left\|  \hat \mu' - \vec P \mu \right\|_2 
&=
\left\| \RowP_{\mathrm{high}}^\top \muh + \mul - 
\lp( \RowP_{\mathrm{high}}^\top \RowP_{\mathrm{high}} + \RowP_{\mathrm{low}}^\top \RowP_{\mathrm{low}} \rp)
\vec P \mu \right\|_2  \\
 &\leq 
 \left\| \RowP_{\mathrm{high}}^\top \muh -   \RowP_{\mathrm{high}}^\top \RowP_{\mathrm{high}} \vec P \mu \right\|_2 + 
 \lp \| \mul -  \RowP_{\mathrm{low}}^\top \RowP_{\mathrm{low}} \vec P \mu \rp\|_2 \tag{the triangle inequality} \\ 
&\leq 
\left\|  \muh -   \RowP_{\mathrm{high}} \vec P \mu \right\|_2 + 
\lp \| \mul -  \RowP_{\mathrm{low}}^\top \RowP_{\mathrm{low}} \vec P \mu \rp\|_2 
\tag{orthonormality of $\vec P_{\mathrm{high}}$}
\\  
&\leq  \frac{1 + \log_2 d}{10 \L} \| \hat{\mu} - \mu\|_2 + \frac{(1+\log_2 d)\f(\alpha,n,\tau^2\alpha)}{10 \L} \tag{\Cref{eq:error1,eq:error2}} \\
&+ O\left( \sqrt{\frac{(d+\log(1/\tau))\varb}{\alpha n}}(1+\log_2 d)   \right)
+ O( \log D \; \varb ) \; \f(\alpha,n) \, ,
\end{align*}
where the first equality follows from the definition of $\hat \mu'$ (\Cref{line:combine}) and the fact that $\RowP_{\mathrm{high}}^\top \RowP_{\mathrm{high}}$ and $\RowP_{\mathrm{low}}^\top \RowP_{\mathrm{low}}$ 
are orthonormal projectors of two complementary subspaces respectively.
Observing that $1 + \log_2 d = \log_2(2d) $ completes the proof of \Cref{lem:single_stage_analysis}.
\end{proof}

We are now ready to complete the proof of \Cref{thm:main}.

\begin{proof}[Proof of \Cref{thm:main}]

    In this proof, we put everything together in order to show that the output of the algorithm \Call{EngangledMeanEstimation}{$N$} (cf. \Cref{alg:mean_estimation} and \Cref{alg:recursive}) satisfies the guarantees of \Cref{thm:main}.
    First, we copy below various quantities that will frequently show up in the argument:
    \begin{align}\label{eq:quantities}
         r = \rr, \; m=\mm, \; t = 2+ m ( 3r + 1), \;  n= N/t, \; \tau = N^{-3}/r.
    \end{align}
    {We note that the sample budget parameter $n$ satisfies the following lower bounds.}
    \begin{claim}
    \label{clm:n-lb}
    {Consider the same context as in the statement of \Cref{thm:main} and also the notations in \Cref{eq:quantities}.
    The conditions in that statement imply that
    $n \gg \log(t/\tau) /\alpha$ and
    $n \gg (D/\alpha) \log(D / \tau) \log(D / (\alpha \tau))$.}
    \end{claim}
    \begin{proof}
    {Recall that in \Cref{thm:main} we assume that $N \gg (D / \alpha) \log^C( D / \alpha )$, 
    where $C$ is some sufficiently large constant.
    }
    {Plugging the definition of $t$ and $\tau$ gives
    \begin{align*}
        \frac{\log(t/\tau)}{\alpha} \lesssim \frac{\log N}{\alpha}.
    \end{align*}}
    {Since $N \geq \log^C(1/\alpha) / \alpha$, \Cref{fact:implicitinequality}, we must have $N \gg \log^3(N) / \alpha$.}
    {Since $n = N/ t = \Omega \lp( N / \log^2(N) \rp)$, 
    it follows that
    $$
    n \gg \log N / \alpha \gtrsim \log(t/\tau) / \alpha.
    $$}

    {Again by \Cref{fact:implicitinequality}, we must have that
    $$
    N \gg (D/\alpha) \log^2(N)
    \log(D / \alpha) \log(D N /  \alpha ) 
    \gtrsim (D / \alpha) \log^2(N) \log( D / \alpha )
    \log(D / (\tau \alpha)).
    $$}
    {This further implies that
    $$
    n \gg (D / \alpha) \log( D / \alpha )
    \log(D / (\tau \alpha))\, ,
    $$
    and concludes the proof of \Cref{clm:n-lb}.}
    \end{proof}
    
    The pseudocode in \Cref{alg:mean_estimation} assumes that the algorithm is able to draw independent datasets according to the data generation model of \Cref{def:model2}. We first argue that this can be simulated in the following claim. 
    \begin{claim}\label{cl:simulation}
        Consider the same context as in the statement of \Cref{thm:main}.
        Then all the datasets that are drawn from the data generating model of \Cref{def:model2} throughout the execution of \Cref{alg:mean_estimation} can be simulated by having access to the model of \Cref{def:model}, and the simulation succeeds with probability at least $1-\tau$.
    \end{claim}
    \begin{proof}
        We start by noting the number and size of the datasets used throughout the execution of \Call{EntangledMeanEstimation}{$N$}.
        First, the algorithm in \Cref{line:initial_guess} uses two datasets of size $n$.
    Then each of the {$r = \rr$} iterations of the loop of \Cref{line:stages} of \Cref{alg:mean_estimation} calls \Call{RecursiveEstimate}{$\vec I,n,0,\tau$}. 
    A single call of \Call{RecursiveEstimate}{$\vec I,n,0,\tau$}  draws one dataset of size $n$  in \Cref{line:dataset} and another two datasets in line \Cref{line:tournament}. Then the algorithm proceeds to call itself recursively for a recursive depth of at most {$m=\mm$}. The final call of \Call{RecursiveEstimate}{$\cdot$} when the recursion ends uses a single dataset of size $n$ (cf. \Cref{line:naive_est}). Thus, in total, the algorithm is using $t = 2+ m ( 3r + 1)$ independent datasets of size $n$ each. This means that the batch size is $n =N/(2+ m ( 3r + 1))$.

    Using \Cref{lem:simulation} with $t = 2+ m ( 3r + 1)$, we have that these datasets can be simulated by access to the model of \Cref{def:model2}, so long as $n \gg \log(t/\tau)/\alpha$, {which is guaranteed by \Cref{clm:n-lb}.}    
    This completes the proof of \Cref{cl:simulation}
    \end{proof}
    
    Given that the simulation succeeds, we now focus on bounding the error of the output.

    We will analyze how the estimation error decreases in each round of the for loop of \Cref{line:stages} of \Cref{alg:mean_estimation} by using \Cref{lem:single_stage_analysis} with $\RowP=\vec I$, $d = D$, {$\delta = 3$}, and $\tau = N^{-3}/r$.
    The lemma is indeed applicable because of the following: (i) $D$ is either a power of 2 or by padding extra coordinates with $\vec 0$ mean we can trivially extend their dimension until it becomes a power of $2$, (ii) the requirement $n \gg \tfrac{d}{\alpha} \log(\tfrac{d}{\tau})\log(\tfrac{d}{\alpha \tau})$ {is guaranteed by \Cref{clm:n-lb}}.
    The application of \Cref{lem:single_stage_analysis} yields the following: If $\hat{\mu}$ is the estimate at the beginning of an iteration of \Cref{line:stages} of \Cref{alg:mean_estimation}, then the improved estimate $\hat \mu'$ after that iteration ends satisfies the following with probability at least $1-O( N^{-0.9})$:\footnote{We rewrite the probability of success listed in \Cref{lem:single_stage_analysis} using 
    {$\tau = N^{-3}/r$, $\delta = 3$, and the fact that
    $\log D < \log N < N^{0.1}$ for sufficiently large $N$.}}
    \begin{align}\label{eq:simplified}
        \| \hat \mu' - \hat{\mu} \|_2 &\leq \frac{1}{2}\|\hat{\mu} - \mu\|_2 + T\;, \\\;\;\text{where} \;\; T &\lesssim \sqrt{\frac{(D+\log(1/\tau))\log(n D/\tau)}{\alpha n}}\log D    +  \log D \log(nD/\tau) f(\alpha,n) \;. \label{eq:capitalT}
    \end{align}

    We now turn our focus on the initial estimate from \Cref{line:initial_guess} of \Cref{alg:mean_estimation}.
    {By \Cref{lem:guess-center},
    the initial estimate $\hat \mu$ has error 
    $O \lp( \sqrt{D} + f(\alpha,n) \rp)$.
    Hence, if the algorithm does not enter \Cref{line:just-return}, we must have 
    $\| \hat \mu - \mu\|_2 \leq O(\sqrt{D})$.
    with probability at least $1 - 1/N$.
    }
    Then \Cref{eq:simplified} implies that after $r = \ceil{0.5\log_2(N)}$ iterations, the error of the estimate $\hat \mu$ produced after the last iteration of the for loop (\Cref{line:return}) must satisfy 
    \begin{align}
    \label{eq:before-final-error}
        \| \hat \mu - \mu \|_2 &\lesssim \sqrt{\frac{D}{N}}  + T\log N \;,
    \end{align}
    with probability at least $1-O( N^{-0.8})$ (where we did a union bound over $r \leq \log N \leq N^{0.1}$ many iterations).
    To get the final error upper bound, we plug in the parameters from \Cref{eq:capitalT,eq:quantities} and the assumption $D \ll N$ into \Cref{eq:before-final-error}. This yields the following error guarantee:
    \begin{align*}
        \| \hat \mu - \mu \|_2 &\lesssim \log^{O(1)}(N)\left(\sqrt{\frac{D}{\alpha N} } + f\left(\alpha, \frac{N}{(\log N)^{O(1)}}\right)       \right)  \;.
    \end{align*}
    Finally, we can further simplify this bound by  the property $f(\alpha,n/L) \leq \poly(L) f(\alpha,n)$ for any $L  = \polylog(N)$. 
    This follows just by the form of $f(\cdot)$ as two pieces of ratios of polynomials that agree on the point where the two pieces meet.
    This allows us to write $f\left(\alpha, \frac{N}{(\log N)^{O(1)}}\right) \lesssim (\log N)^{O(1)}f\left(\alpha, N\right)$, which yields the final bound appearing in the statement of \Cref{thm:main}. 
    {Regarding the probability of failure, in this proof we used $\delta =3 $ and we ended up with probability of failure $1-O( N^{-0.8})$, which is at least $0.99$ as long as $N$ is sufficiently large.}
    This concludes the proof of \Cref{thm:main}.
\end{proof}

\newpage

\printbibliography

\newpage
\appendix

\section*{Appendix}

\section{Additional Preliminaries}\label{sec:additional-prelims}

We restate the definitions of the two models below.

\ORIGINALMODEL*
\MULTIBATCHMODEL*
We are now ready to show that the above data-generation model can be simulated efficiently with a single dataset from the data-generation model of \Cref{def:model}.
\SIMULATION*
\begin{proof}
To simulate access to $t$ batches of samples, each of size $n$, from the data generation model of \Cref{def:model2}, we set $N = t n$ and draw $N$ samples from the data generation model of \Cref{def:model}. 
We will randomly permute the samples and then partition the samples into batches of size $n$ in order.
It remains to argue that the signal-to-noise-rate of each batch is at least $0.9 \alpha$ with probability at least $1 - \tau$. 
We focus on the first batch as the arguments for the rest are identical.
Denote by $g_i$ for $i \in [N]$ the indicator variable that equals to $1$ when the $i$-th sample (after the random permutation) has its covariance bounded from above by $\vec I$.
It then suffices to argue that $\sum_{i=1}^{n} g_i \geq 0.9 \alpha n $.
On the one hand, by linearity of expectation, we have that
\begin{align}
\label{eq:noise-expectation}
\E\left[ \sum_{i=1}^{n} g_i \right] = \sum_{i=1}^{n} \E[ g_i ] = \alpha n.    
\end{align}
On the other hand, by the main result of~\cite{joag1983negative}, since $g_1, \cdots,  g_{N}$ follow the permutation distributions, we have that $g_1, \cdots, g_{n}$ enjoy negative association (Definition 2.1 of  \cite{joag1983negative}).
We can therefore apply the Chernoff-Hoeffding bounds for negatively-associated variables (\cite{dubhashi1996balls}), which yields that
\begin{align}
\label{eq:na-concentration}
\Pr\lp[  \sum_{i=1}^{n} g_i  
\leq 0.9 \E \lp[ \sum_{i=1}^{n} g_i \rp] 
\rp] \leq \exp\lp( - \Omega\lp( \E \lp[ \sum_{i=1}^{n} g_i \rp] \rp) \rp).
\end{align}
Combining \Cref{eq:noise-expectation,eq:na-concentration} hence gives that 
$\sum_{i=1}^{n} g_i \geq 0.9 \alpha n $ 
with probability at least $\tau / t$
as long as  $n \gg \log(t/\tau) / \alpha$.
It then follows from the union bound that the signal-to-noise-rate for each batch is at least $0.9 \alpha$ with probability at least $1 - \tau$.
\end{proof}

\end{document}

%% file: macro.tex
\newcommand{\lp}{\left}
\newcommand{\rp}{\right}

\newcommand{\Mpop}{\vec {\tilde M}_{(\text{pop})}}

\newcommand{\accept}{\mathcal A}
\newcommand{\recursivefunction}{\textsc{RecursiveEstimate}}

\newcommand{\Memp}{\vec {\tilde M}_{(\text{emp})}}

\newcommand{\meanini}{{\tilde \mu}}  %
\newcommand{\muh}{\mu_{\mathrm{high}}}
\newcommand{\mul}{\mu_{\mathrm{low}}}

\newcommand{\Sigmavg}{  \vec {\tilde\Sigma}_{\mathrm{avg}} }

\newcommand{\Sacc}{ S_{\mathrm{accept}} }
\newcommand{\Snice}{ \mathcal{S}_{\mathrm{nice}} }

\newcommand{\mean}{\mu}

\newcommand{\normb}{4d\log(nd/\tau)}
\newcommand{\anormb}{d\log(nd/\tau)}

\newcommand{\varb}{ \log(nd/\tau) }

\usepackage{mathtools}
\DeclarePairedDelimiter\ceil{\lceil}{\rceil}
\DeclarePairedDelimiter\floor{\lfloor}{\rfloor}

\newcommand{\RowP}{\vec P}

\renewcommand{\L}{\log D}
\renewcommand{\hat}{\widehat}
\renewcommand{\Sacc}{ S_{\accept} }

\newcommand{\rr}{ \ceil{\log_2 N} }
\newcommand{\mm}{ \log_2 D }

\newcommand{\rejmean}{\tilde \mu}

\newcommand{\f}{f_\delta}

\newcommand{\pnormb}{2 d \log(n/\tau)}

\newcommand{\SearchSubSpace}{\textsc{FindSubSpace}}
\newcommand{\PartialEstimate}{\textsc{PartialEstimate}}